%% file: master.tex
\documentclass[11pt]{article}
\pdfoutput=1
\usepackage{amsmath}
\usepackage{amssymb}
\usepackage{theorem}
\usepackage{epsfig}
\usepackage{url}
\DeclareSymbolFontAlphabet{\Bbb}{AMSb}

\newtheorem{theorem}{Theorem}[section]
\newtheorem{proposition}[theorem]{Proposition}
\newtheorem{lemma}[theorem]{Lemma}
\newtheorem{corollary}[theorem]{Corollary}
\theorembodyfont{\upshape}
\newtheorem{definition}[theorem]{Definition}
\newtheorem{example}[theorem]{Example}
\newtheorem{remark}[theorem]{Remark}
\theorembodyfont{\upshape}

\input{symbolne}

\begin{document}

\author{  Clint Scovel\thanks{Current contact information: California Institute of Technology,
clintscovel@gmail.com.}\\
Los Alamos National Laboratory\\
{\it jcs@lanl.gov}
\and
 Ingo Steinwart\\
 Universit\"{a}t Stuttgart\\
{\it Ingo.Steinwart@mathematik.uni-stuttgart.de }
\\[5mm]
}
\title{Hypothesis Testing for Validation and Certification
}

\maketitle


\begin{abstract}

We develop a hypothesis testing framework for the formulation 
 of the problems of 1) the {\em validation} of a simulation model and
2) using modeling to  {\em certify}  the performance of a physical system\footnote{This document is essentially an exact copy of one dated April 12, 2010.}.
These results are used to solve the extrapolative validation and certification problems, namely problems where the regime of interest is different than the regime for which we have experimental data.
 We use concentration of measure theory to develop the tests
and analyze their errors. This work was stimulated by the work of
 Lucas, Owhadi,
and Ortiz
\cite{LuOwOr08a} where  a rigorous method of 
 validation and certification  is described and
tested.
In Remark \ref{rem_ortiz} we describe the connection
between the two approaches. Moreover, as mentioned  in that work these
results have important implications in the Quantification of Margins and
Uncertainties (QMU) framework. 
In particular, in Remark \ref{rem_qmu} we describe how
it provides a rigorous interpretation of the notion of {\em confidence}
and  new notions of {\em margins} and
{\em uncertainties} which allow this interpretation.
Since certain {\em concentration parameters} used in the above tests may be
unkown, we furthermore show, in the last half of the paper, how to derive equally powerful tests which estimate them from sample data, thus
replacing the assumption of the values of the concentration parameters
with weaker assumptions. 
\end{abstract}

\input{vv}

\bibliographystyle{unsrt}

\bibliography{np}

\end{document}

%% file: symbolne.tex
\newcommand{\PR}{\Bbb{P}}
\newcommand{\R}{\Bbb{R}}    

\newcommand{\E}{\Bbb{E}}    
\newcommand{\F}{\Bbb{F}}

\newcommand{\rem}[1]{}

\def \a         { \alpha }

\def \e         { \varepsilon }
\def \r         { \rho }

\def \t         { \tau }

\def \d         { \delta }

\def \x         { \xi }

\newcommand{\Defi}{\ := \ }

\newlength{\fixboxwidth}
\newcommand{\BlackBox}{\rule{1.5ex}{1.5ex}}  
{\end{list}}

\newenvironment{proof}{\begin{list}{{\bf \em Proof: }}%
{\setlength{\labelsep}{0pt}\setlength{\leftmargin}{0pt}\setlength{\labelwidth}{0pt}}\item}%
{\hfill\BlackBox\end{list}}
{\hfill\BlackBox\end{list}}
\newenvironment{proofof}[1]{\begin{list}{{\bf \em Proof of #1: }}%
{\setlength{\labelsep}{0pt}\setlength{\leftmargin}{0pt}\setlength{\labelwidth}{0pt}}\item}%
{\hfill\BlackBox\end{list}}
{\hfill\BlackBox\end{list}}

\newcommand{\snorm}[1] {\Vert #1 \Vert}

\newcommand{\sLp}[1]{\mbox{${\cal L}_p(\mu)$ }}

%% file: vv.tex
\section{Introduction}

Validation of simulation models is clearly important and much substantial work has
been directed towards it, see e.g.~\cite{ObTr02,ObTr04,Eardley,Balci,Sargent96,
Sargent04,Sargent05a} and the references therein.
Moreover, the problem appears to
go straight to the heart of the philosophy of science (see e.g.~\cite{NaFi,KeGa,KlOnGa}).
Indeed,  \cite{OrShBe} assert that validation is impossible, and \cite{LuOwOr08a} describe
a rigorous
 method for it. 
 On the other hand, it appears
that while all agree that validation is an important and difficult problem, few agree on
what the problem actually is. In the words of G. K. Chesterton  \cite[pg.~ix]{BarPre86},
"It isn't that they can't see
the solution. It is that they can't see the problem." 
  In this paper we
formulate examples of both 
the problems of validation and certification as problems of
constructing hypothesis
tests. A straightforward analysis using concentration of measure theory then provides tests and guarantees on their performance.

  Although hypothesis tests have been used in validation before, e.g.~in
\cite{BaSa,Kleijnen}, our formulation is quite different.  In particular, we formulate
null and alternate hypotheses which represent a flexibility in the customer's
specification of a performance design threshold. We develop tests  
that require a clear delineation of  assumptions
 and  then use concentration of measure inequalities to analyze the performance of
the tests. These results are then used to solve the extrapolative validation and
certification problems, namely problems where the deployment regime is different
than the experimental regime.
This framework is then compared with that of Lucas, Owhadi and Ortiz \cite{LuOwOr08a}.  As mentioned in that work,
these results also have important implications in the Quantification of Margins and
Uncertainties (QMU) framework discussed in detail in \cite{QMU09,PiTrHe06,ShWo03}. In
particular, in Remark \ref{rem_qmu} we 
discuss how these results provide  a rigorous interpretation of the notion of {\em confidence} and a new notion of 
{\em uncertainties} which allow this interpretation.
Since certain {\em concentration parameters} used in the above tests may be
unknown, we furthermore show how to derive equally powerful tests which estimate them from
sample data, thus
replacing the assumption of the values of the concentration parameters
with much weaker assumptions.
This humble
beginning needs to be refined so that it fits better with real applications. It should also
incorporate some of the conclusions and structure of the above-mentioned works, but we leave that for the future.
 Let us now describe our framework and formulations.
At a high level we say that {\em validation} is the assessment of the quality of a model of a
physical system and {\em certification} is using modeling to assess the performance of a physical
system. To make these notions more specific, consider the following general framework which
will be used for both validation and certification.

Consider the case of a real-valued random variable $U$ that
describes the performance of a system and  a customer who would like to have a
quantitative guarantee on this performance. You inform the customer that
you can consider a test of the hypothesis 
\[ \PR\bigl(U \geq a\bigr) \geq p\]
where $a$ is the performance design threshold and $p$ is a level of
confidence. When pressed to provide the specific values of the
parameters $a$ and $p$ the customer may provide values, for example
$a=1000$ and $p=.95$. However, if you then ask him whether $a=950$ and
$p=.93$ would be acceptable, he might respond in the affirmative.
Consequently, a more realistic test might be to test
\begin{equation}
\label{htuncertain}
 \PR\bigl(U \geq A\bigr) \geq P
\end{equation}
where $A$ and $P$ are sets instead of real numbers. However, what
(\ref{htuncertain}) actually means and how to construct and analyze a
test for it are not clear. To resolve this problem, let 
 us introduce some notation.
Let ${\cal U}$ denote the set of real-valued
 random variables. 
For $a \in \R,\, p \in (0,1)$ define the null hypothesis by
 \[{\cal H}_{a,p} \Defi
\bigl \{U \in {\cal U}: \PR(U \geq a) \geq p \bigr \}\]
and the alternative by
\[
 {\cal K}_{a, p} \Defi \bigl \{U \in {\cal U}: \PR(U \geq a) <
p \bigr \}.\]
Consider $a' \leq a$,  $p' \leq p$ and suppose that
 $U \in  {\cal H}_{a,p} \cap {\cal
K}_{a',p'}$. Then, since
\[p' > \PR(U \geq a')\geq  \PR(U \geq a) \geq
p \]  is a contradiction,
 we conclude that
\begin{equation}
\label{basic} {\cal H}_{a,p} \cap {\cal
K}_{a',p'}=\emptyset ,\quad a' \leq a,\, p'
\leq p \, .
\end{equation}
Therefore, when $a' \leq a$ and $ p' \leq p$ we can consider
a test of   ${\cal H}_{a,p}$ against
${\cal K}_{a',p'}.$
Now let $a$ and $p$  be specified and specify  tolerance intervals
 $A$ and $P$ such 
that $A \leq a$ and $P \leq p$, where the notation implies that $a \in
A$ and $p \in P$. Then by (\ref{basic}) we can define
a test of  (\ref{htuncertain}) by testing ${\cal H}_{a,p}$ against
${\cal K}_{a',p'}$ for some $a' \in A$ and $p' \in P$. 	Given
the  freedom the tolerance intervals  allow in the choice of $a'$ and
$p'$, we seek to choose them to our advantage.

Let us first consider the case where $A=\{a\}$, namely there is no
tolerance to changing the design criterion. We wish to construct a test
of  ${\cal H}_{a,p}$ against
${\cal K}_{a,p'}$ for $p' \in P$. Let $U_{i}, i=1,..,n$ be i.i.d.~samples
from $U$. We can form a test by composing the sample data
$U_{i}, i=1,..,n$ with the indicator function  $I_{a}: \R
\rightarrow  \{0,1\}$
defined by $I_{a}(u)=1, u \geq a$ and $I_{a}(u)=0, u  < a$ to obtain
Bernoulli random variables $I_{a}\circ U_{i}$. That is, we 
simply evaluate whether the sample points are
 greater than or equal to $a$ or not.
We form a test of  ${\cal H}_{a,p}$ against
${\cal K}_{a',p'}$
by forming the binomial test of ${\cal H}_{p}$ against
${\cal K}_{p'}$ where
$${\cal H}_{p}\Defi \{X: \PR(X=1)=r, \, \PR(X=0)=1-r,\, r \geq p\}$$
and
$${\cal K}_{p'}\Defi \{X: \PR(X=1)=r,\,  \PR(X=0)=1-r,\, r <p'\}.$$
By the Neyman-Pearson Lemma  \cite[Thm.~3.1]{Lehmann94} and
\cite[Thm.~3.2]{Lehmann94} we know there exists  a uniformly most
powerful test of ${\cal H}_{p}$ against
${\cal K}_{p'}$  (see e.~g.~ \cite[Ch.~3]{Lehmann94}).
However, this uniformly most powerful test is
characterized through the
binomial distribution. The statement of approximate tests with rigorous
guarantees on their type I and II errors appears, in principle,
 to be available but
evidently it is no easy matter. Rigorous bounds connecting the binomial
distribution to the normal can be found in  Feller  \cite{Feller45}
and to the Poisson distribution in Anderson and Samuels \cite{AnSa65}. Guarantees
outside of the range of applicability of these results can be found in
 Slud \cite{Slud77}. Approximations to the optimal test parameters have been
derived and studied empirically in Shore \cite{Shore86a,Shore86b} and
 Chernoff
\cite{Chernoff52} has analyzed the  asymptotics, in
particular
when
$p'$ is close to $p$. 
 Although a comprehensive rigorous analysis of this case should be completed,
that is not our goal here. 
Instead we consider the  case where $ P=\{p\}$, where there is no tolerance to the value $p$, 
but a nontrivial tolerance in the design criteria $A$. That is, we
test ${\cal H}_{a,p}$ against
${\cal K}_{a',p}$ for some $a' \in A$.  For simplicity
we remove the $p$ from the notation of the
hypothesis spaces, that is, from now on ${\cal H}_{a,p}$ and ${\cal
K}_{a',p}$
are denoted by ${\cal H}_{a}$ and ${\cal K}_{a'}$  respectively.
We will show that
reducing the spaces of random variables further allows the
development and analysis of efficient tests and that this analysis is
quite elementary. The full problem of testing ${\cal H}_{a,p}$ against
${\cal K}_{a',p'}$ for  $a' \in A$ and $p' \in P$ where both
tolerance intervals are nontrivial might be accomplished through
 a combination of the
above mentioned analysis and the results herein. 
To reduce the null and alternative
hypothesis sets we will consider random variables $U$ which are
generated as $U=F(X)$ by real functions $F:X \rightarrow \R$ where each $X$
is a vector random variable. We make assumptions on this set of
functions and vector random variables that guarantee the
degree of concentration of $U$ about its mean in terms 
of a concentration parameter ${\cal D}$ (all this will be clarified below). We
denote by ${\cal U}_{{\cal D}}$  the
resulting space of real-valued random variables 
and reduce the null and alternative hypothesis spaces
accordingly.    
 Having performed this reduction,
 we will demonstrate how to construct tests of
$\PR(U\geq A) \geq p$ in terms of ${\cal D}$ and describe their type I and II errors.
In
addition, we observe that if $A$ is large enough compared to ${\cal D}$ we can obtain 
 tests with small type I and II errors. We disregard measurability
 considerations.
In many applications, we want to validate a model or certify a physical system
 in the deployment regime where the real physical system is impossible or expensive to sample. In Section \ref{sec_xtrap} we obtain the first results, as far as we can tell,
 for this extrapolation problem.

To apply these results to validation, we let $U=F(X)$ denote a measure of a model's fit to a physical system with respect to
a quantity of interest. For example if,
for the value $x$ of the random variable $X$, the model predicts the strength of a material to be
$s_{M}(x)$ and the physical system obtains the strength $s_{Ph}(x)$ then we might define
$F(x):= \frac{1}{|s_{M}(x)-s_{Ph}(x)|}$. Then surpassing the performance threshold $a$ is equivalent
to $|s_{M}(x)-s_{Ph}(x)| \leq \frac{1}{a}$. We apply the above mentioned result to obtain a solution to the validation problem of constructing a test
of ${\cal H}_{a}$ against ${\cal K}_{a'}$  using samples from $F(X)$ which has small type I and II
errors.  
To apply this result to certification, we let $F(X)$  be the performance
 of the physical system and $M(X)$ be the performance of the physical system predicted by the model.
 For example, let $F(x):= s_{Ph}(x)$ be the strength  
of the physical system and $M(x):= s_{M}(x)$ be the strength of the physical system simulated by 
the model.  We apply the above mentioned result to obtain a solution to the certification problem of constructing a test of ${\cal H}_{a}$ against ${\cal K}_{a'}$
using samples from $F(X)$ and $M(X)$ which has small type I and II
errors.
Using the above mentioned tests, we observe in a quantitative way the intuitive
result that
 if the validation diameter
$ {\cal D}_{F-M}$   is much smaller than the model diameter ${\cal D}_{M}$,
 then we need
 much fewer samples of the real physical system $F$ than the
model $M$ to certify the performance of $F$.
In Remark \ref{rem_ortiz} we describe the connection to the rigorous validation and
certification results 
 of Lucas, Owhadi, and Ortiz \cite{LuOwOr08a}.
These results generalize easily to other concentration
inequalities. In particular, using concentration theorems for non $i.i.d.$ sampling we can, with a
substantial increase in complexity, obtain good tests when the empirical data are not generated
$i.i.d.$ or when the components of the random vector $X$ are not independent. 
These tests and bounds on their performance
require knowing the  values of the diameter ${\cal D}_{F}$  for validation
and ${\cal D}_{F-M}$ and
${\cal D}_{M}$ for certification.
Since good approximations to these values may not be known in practice, 
we show, beginning in Section \ref{sec_est}, how to estimate them
to derive equally powerful tests,
replacing the assumption of the values of the concentration parameters
with much weaker assumptions. These tests provide validation and certification tests with
{\em estimated diameters}.

\section{Validation and Certification with Known Diameters}

Let us first describe the concentration parameter ${\cal D}$ mentioned above,
Let $(\Omega,{\cal F}, \PR)$ denote a probability space and consider a
product space ${\cal X}={\cal X}^{1}\times\cdots \times {\cal X}^{m}$.
 We call  a mapping
$X:\Omega \rightarrow {\cal X}$ a random vector with
range ${\cal X}$  and will abuse notation by also using the symbol $\PR$
 for
the image probability measure on ${\cal X}$.
For a function $F: {\cal X} \rightarrow \R$ we
define the partial diameters to be
\begin{equation}
\label{def_diameter}
  D^{F}_{j} = \sup_{x_{k}=x'_{k},\, k \neq j}
{ \bigl( F(x) -F(x') \bigr)} \quad j=1,..,m
\end{equation}
where the supremum is taken over all $x, x' \in {\cal X}$
 which differ only in their $j$-th component.
Let  $ {\cal D}_{F}^{2}
:= \sum_{j=1}^{m}{( D_{j}^{F})^{2}}$ define the McDiarmid diameter ${\cal
D}_{F}$ of
the function
$F$.
For a vector random variable $X$ and function $F:{\cal X}
\rightarrow \R$ we consider the random variable $ 
 F \circ X: \Omega \rightarrow \R$ which we also denote by $F$.
For  the  random variable $F$ we have McDiarmid's inequality
\cite[Thm.~3.1, pg.~206]{McDiarmid91}:
\begin{theorem}
\label{conc_mcdiarmid}
Let ${\cal X}={\cal X}^{1}\times\cdots \times {\cal X}^{m}$ be  a
Cartesian product and
let $F:{\cal X} \rightarrow \R$ have the McDiarmid diameter $ {\cal D}_{F}$.
  Then for any product probability measure $\PR=\mu_{1} \otimes
\cdots
\otimes \mu_{m}$ we have
\[ \PR(F - \E F \geq r ) \leq e^{-\frac{2r^{2}}{{\cal D}_{F}^{2}}}\]
\end{theorem}
 If, for
 $ 0 < t < 1$, we define
\[r_{t}  \Defi \frac{{\cal D}_{F}}{\sqrt{2}}  \sqrt{ \log{t^{-1}}}\, \]
then we have the following useful inequalities:
\begin{eqnarray*}
\label{rt}
&\PR(F - \E F \geq  r_{t})& \leq t\\
& \PR(F - \E F >  r_{t}) & < t\\
&\PR(F - \E F \leq  -r_{t})& \leq t\\
& \PR(F - \E F <  - r_{t}) & < t . 
\end{eqnarray*}
Since this theorem's only dependence on $F$ and $X$ is through the
parameter $ {\cal D}_{F}$ we can define the subset ${\cal U}_{{\cal D}} \subset {\cal U}$
consisting of all real-valued random variables generated as $U=F(X)$
for some $F$ and $X$
such that ${\cal D}_{F} \leq  {\cal D}.$
Let
  ${\cal
H}^{{\cal D}}_{a} := {\cal H}_{a} \cap {\cal U}_{{\cal D}}$ and  ${\cal
K}^{{\cal D}}_{a}
:= {\cal K}_{a}\cap {\cal U}_{{\cal D}}$ denote null and alternative
generated
in this way and consider testing ${\cal
H}^{{\cal D}}_{a}$ against ${\cal K}^{{\cal D}}_{a'}$. 
We are now ready to state our main result which we then use to establish both validation and
certification results.
We describe a test of ${\cal H}^{{\cal D}}_{a}$ against
${\cal K}^{{\cal D}}_{a'}$ for  $a-a'$ bounded below in terms of ${\cal D}$ and $p$.
Therefore if $[a' ,a] \subset A$, then
the following result provides   a test of $\PR\bigl(U \geq A\bigr) \geq p$, with
bounds on its errors.
 Note that the test is in terms of the value of
a function  $F':Y \rightarrow \R$ for some random vector $Y$ with the
only constraint being $\E F' =\E U$. All tests in this work accept
the null ${\cal H}^{{\cal D}}_{a}$ by producing $T=1$ and reject otherwise. Recall that the
type I error is defined by $\theta_{1}(U) :=
\PR(T=0), U \in {\cal H}^{{\cal D}}_{a}$
and the type II error is defined by $\theta_{2}(U) :=
\PR(T=1), U \in {\cal
K}^{{\cal D}}_{a'}. $

\begin{theorem}
\label{thm_main}
Let $ 0 < p <1$,  ${\cal D}, {\cal D}'> 0$ and consider $ U \in {\cal U}_{{\cal D}}$. Moreover, consider a vector random
variable $Y$  and a function 
 $F':Y
\rightarrow \R$ with diameter ${\cal D}_{F'} \leq {\cal D}'$
such that $\E F'=\E U$.
For  $0 < t <1$
define
$r_{t} :=  \frac{ {\cal D}}{\sqrt{2}} \sqrt{   \log{t^{-1}}}$
and $r'_{t} := \frac{  {\cal D}'}{\sqrt{2}} \sqrt{   \log{t^{-1}}}$.
Let
   $ 0 < \d_{1}, \d_{2} < 1,$ 
and   let $a$ and $a'$ satisfy
 $a-a' \geq r_{p} +r_{1-p} +r'_{\d_{1}}+r'_{\d_{2}}$ so that the interval $[a'+r_{1-p}
 +r'_{\d_{2}},
  a-r_{p} -r'_{\d_{1}}]$ is nonempty. Let $b \in
[a'+r_{1-p} +r'_{\d_{2}},
  a-r_{p} -r'_{\d_{1}}].$
Then the  test $T$ of  ${\cal H}^{\cal D}_{a}$  against ${\cal K}^{\cal D}_{a'}$
defined by
\[T \Defi
\begin{cases} 1, & F'(y) \geq b\\
              0, & F'(y) <b \, 
\end{cases}
\]
satisfies
\[  \theta_{1}(U)  < \d_{1}\, ,
\]
\[\theta_{2}(U) 
\leq \d_{2} \, .\]
\end{theorem}

The condition $\E F' =\E U$ of Theorem \ref{thm_main} 
can be easily satisfied when i.i.d.~samples are
available. Therefore, in this case, it is straightforward to use
Theorem
\ref{thm_main} to define  tests, with guarantees on their errors, for both {\em validation} and {\em certification}.
\begin{corollary}[Validation]
\label{cor_val}
Let $U=F(X)$ and suppose ${\cal D} \geq {\cal D}_{F}$. 
Let $F(X_{i}),i=1,..,n$ be i.i.d. samples of $F(X)$
and define
  $ \langle F\rangle_{n}
 := \frac{1}{n}\sum_{i=1}^{n}F(X_{i})$ to be the sample mean.
Let $0 < p,\d_{1},\d_{2} < 1$  and for $0 < t < 1 $
define
$r_{t}  := \frac{{\cal D}}{\sqrt{2}} \sqrt{   \log{t^{-1}}}$.
Moreover,  let $a$ and $ a'$ satisfy
 $a-a'\geq r_{p} +r_{1-p} +\frac{1}{\sqrt{n}}r_{\d_{1}}
+\frac{1}{\sqrt{n}}r_{\d_{2}}$ so that the interval $[a'+r_{1-p} +\frac{1}{\sqrt{n}}r_{\d_{2}},  
  a-r_{p} -\frac{1}{\sqrt{n}}r_{\d_{1}}]$ is nonempty.
Let $b\in [a'+r_{1-p} +\frac{1}{\sqrt{n}}r_{\d_{2}}, 
  a-r_{p} -\frac{1}{\sqrt{n}}r_{\d_{1}}]$ and 
consider the test
 $T$ of ${\cal H}^{{\cal D}}_{a}$ against ${\cal K}^{{\cal D}}_{a'}$
 defined by
\[T \Defi
\begin{cases} 1\, ,  &   \langle F\rangle_{n}
 \geq b\\
              0 \, ,  &
\langle F\rangle_{n} <b \, .
\end{cases}
\]
Then we have
\[  \theta_{1}(U)  < \d_{1}\, ,
\]
\[\theta_{2}(U) 
\leq \d_{2} \, .\]
\end{corollary}

As discussed in the introduction, if  $F(X)$ represents a physical
system and 
 $M(X)$ a model of that system  we can consider how to
 test the performance
of $F$ by
decomposing $F(X) =M(X)+\bigl(F(X)-M(X)\bigr)$ into the model
component and the model deviation component. If the test accepts
we obtain certification.
We now show how to sample the model and the model deviation to test
the performance of the physical system. In particular,  
 the following result  shows
that if the validation diameter  $ {\cal D}_{F-M}$   is much smaller then the model diameter ${\cal D}_{M}$,
 then we need
 much fewer samples of the real physical system $F$ than the
model $M$ to certify the performance of $F$. It is phrased in terms of a general
decomposition $F=F_{1}+F_{2}$.

\begin{corollary}[Certification] 
\label{cor_cert}
Let $U=F(X)$ where
 $F := F_{1}+F_{2}$ is the sum of two functions with diameters
${\cal D}_{F_{1}}$ and ${\cal D}_{F_{2}}$. Let ${\cal D}, {\cal D}_{1}$, and ${\cal
D}_{2}$ satisfy $ {\cal D}_{1}\geq {\cal D}_{F_{1}}$,
${\cal D}_{2}\geq {\cal D}_{F_{2}}$  and  ${\cal D} \geq 
{\cal D}_{1}+{\cal D}_{2}$.
Let $F_{1}(X_{i}),i=1,..,n_{1}$ be i.i.d. samples of $F_{1}(X)$
and define
  $ \langle F_{1}\rangle_{n_{1}}
 := \frac{1}{n_{1}}\sum_{i=1}^{n_{1}}F_{1}(X_{i})$ to be the sample mean. Also let
$F_{2}(X_{i}),i=n_{1}+1,.., n_{1}+n_{2}$ be i.i.d. samples of $F_{2}(X)$
and define
  $ \langle F_{2}\rangle_{n_{2}}
 := \frac{1}{n_{2}}\sum_{i=n_{1}+1}^{n_{1}+n_{2}}F_{2}(X_{i})$ to be the sample mean of
the second set of samples.
For $0 < t < 1 $ define
\[
\r_{t}  \Defi
\frac{1}{\sqrt{2}}\sqrt{\frac{{\cal D}_{1}^{2}}{n_{1}}+\frac{{\cal D}_{2}^{2}}{n_{2}}}
\sqrt{   \log{t^{-1}}}
\]
and
\[r_{t}  \Defi
\frac{{\cal D}}{\sqrt{2}}
\sqrt{   \log{t^{-1}}}
\]
Let  $0 < \d_{1},\d_{2} < 1$ and suppose that $a$ and $a'$ satisfy
 $a-a'\geq
r_{p} +r_{1-p}+\r_{\d_{1}} +\r_{\d_{2}}$ so that the interval $[a'+r_{1-p} +\r_{\d_{2}}, a-
r_{p} -\r_{\d_{1}}]$ is nonempty.  Let   $b \in  
[a'+r_{1-p} +\r_{\d_{2}}, a-
r_{p} -\r_{\d_{2}}]$ and
consider the test
 $T$ of ${\cal H}^{\cal D}_{a}$ against ${\cal K}^{\cal D}_{a'}$
 defined by
\[T \Defi
\begin{cases} 1\, ,  &   \langle F_{1}\rangle_{n_{1}}+ \langle
F_{2}\rangle_{n_{2}}
 \geq b\\
              0 \, ,  &  
\langle F_{1}\rangle_{n_{1}}+ \langle F_{2}\rangle_{n_{2}} <b \, .
\end{cases}
\]
Then $U \in {\cal U}_{{\cal D}}$ and   
 we have
\[  \theta_{1}(U)  < \d_{1}\, ,
\]
\[\theta_{2}U) 
\leq \d_{2} \, .\]

Moreover, suppose ${\cal D}_{2}\leq
{\cal D}_{1}$,   $n_{2} \geq
\frac{{\cal D}_{2}}{{\cal D}_{1}} n_{1}$ and  define 
 \begin{equation}
\label{r_def2}
\r_{t}  \Defi \frac{{\cal D}_{1}}{\sqrt{n_{1}}}
\sqrt{  \log{t^{-1}}}\, 
\end{equation}
in the above conditions on $a,a'$ and the test parameter $b$.
Then for any ${\cal D} \geq 2{\cal D}_{1}$ we have  $U \in {\cal U}_{{\cal D}},$  and
 \[  \theta_{2}(U)  < \d_{1}\, ,
\]
\[\theta_{2}(U) 
\leq \d_{2} \, .\]
\end{corollary}

\begin{remark}[Connection with Lucas, Owhadi and Ortiz \cite{LuOwOr08a}]
\label{rem_ortiz}
If in  Theorem \ref{thm_main} we define 
$b :=  a'+r_{1-p} +r'_{\d_{2}}$ in terms of a performance
value $a'$ 
 the condition for acceptance
in Theorem \ref{thm_main} can be
written 
$ \langle F \rangle_{n}
-a'-r'_{\d_{2}} \geq r_{1-p}$ 
which amounts to
\[\frac{ \langle F \rangle_{n}
-a'-\frac{  {\cal D}'}{\sqrt{2}} \sqrt{   \log{{\d_{2}}^{-1}}}}{{\cal D}}
\geq 
\sqrt{\frac{1}{2}\log{(1-p)^{-1}}}\, .\]
For the case of validation in Corollary \ref{cor_val} we have
${\cal D}_{F'} \leq \frac{1}{\sqrt{n}}{\cal D}_{F}$ and so with $p=1-\epsilon$, $\d_{2}
=\epsilon'$,
${\cal D}:= {\cal D}_{F}$, and ${\cal D}':=  \frac{1}{\sqrt{n}}{\cal D}_{F}$ the test of Corollary
\ref{cor_val}  amounts  essentially (with a $\sqrt{2}$ better multiplicative
factor in last term on the left) to the
 validation criterion of
Lucas, Owhadi, and Ortiz
\cite[Eqn.~40]{LuOwOr08a}
for the exact model with single
performance measure (their Scenario 3). 
For certification, 
 we can apply
 Corollary \ref{cor_cert} with the
choice  $F_{1}$ as their
$F$ and $F_{2}$ as their $G-F$.  Using the inequality
 ${\cal D}_{F_{1}+F_{2}} \leq  {\cal D}_{F_{1}} +{\cal D}_{F_{2}}$ and
setting $n_{1}=n_{2}$ and $\d_{2} =2\epsilon'$  
we again obtain essentially (in a similar way as mentioned above)
the certification criteria  of
\cite[Eqns.~58\&59]{LuOwOr08a}. 
Moreover,   Corollaries \ref{cor_val} and \ref{cor_cert}
 show we can interpret the certification criteria of
\cite{LuOwOr08a} as  guarantees that the type II error is less than
$\d_{2}$.    If we then select $a$ such that $a -a' \geq r_{p} +r_{1-p}
+\r_{\d_{1}}+\r_{\d_{2}}$ we can also assert that the type I
error is less than $\d_{1}$. In particular, if the design parameter
value $a'$ can tolerate being moved  so that
 $ a - a' \geq   r_{p} +r_{1-p}
+\r_{\d_{1}}+\r_{\d_{2}}$ with $\d_{1}$ and $\d_{2}$ small, this criterion
amounts to a hypothesis test with type I and II errors bounded by $\d_{1}$ and $\d_{2}$
respectively. In this sense the criteria of
\cite{LuOwOr08a} appear to correspond with our hypothesis test but with
the roles of the hypothesis spaces  ${\cal H}_{a}$  and $ {\cal K}_{a'
}$ reversed.
\end{remark}

\begin{remark}[Connection with QMU]
\label{rem_qmu}
For a detailed discussion of the QMU framework please see \cite{QMU09,PiTrHe06,ShWo03}.
In the QMU framework, the confidence is evaluated in terms of a ratio $\frac{M}{U}$ where
$M$ is a margin and $U$ is an uncertainty. The National Research Council of the National
Academies report \cite[Finding 1-1]{QMU09} states that
"QMU is a sound and valuable framework that aids the assessment and evaluation of the
confidence in the nuclear weapons stockpile." However it also states
 "There are
serious and difficult problems to be resolved in uncertainty quantification, however,
including the physical phenomena that are modeled crudely or not at all, the possibility of
unknown unknowns, lack of computing power to guarantee the convergence of codes, and
insufficient attention to validating experiments. Finally, they state that
{\em "Even if the uncertainties arising from all of the different sources were estimated, their
aggregation into an overall uncertainty for a given quantity of interest is a problem that
needs further attention."}
Although we do not suggest that we can answer all these question now we can make some
conclusions along these lines using  the discussion of Remark \ref{rem_ortiz}.
For validation,  
 consider the ratio
\[\frac{ \langle F \rangle_{n}
-a'}{{\cal D}}
\] where the numerator is a "margin" and the denominator is an "uncertainty". The
inequality 
\[\frac{ \langle F \rangle_{n}
-a'}{{\cal D}}
\geq 
\sqrt{\frac{1}{2}\log{(1-p)^{-1}}} +\frac{1}{  \sqrt{2n}} \sqrt{
\log{{\d_{2}}^{-1}}} \]  shows two
things. First it shows how we can interpret confidence. That is,
if $\PR(F \geq a') < p$,
namely if the performance is insufficient, then with
probability less than $\d_{2}$ will we accept the performance as sufficient.
Namely, our confidence is $\d_{2}$. Moreover, the precise definition of the uncertainty
parameter ${\cal D}$ shows how this parameter is aggregated so as to
maintain the interpretation of the confidence statement. For the certification problem
similar comments also apply but we get the added benefit of seeing how modeling
uncertainties and validation uncertainties are aggregated and combined and how they
influence the number of validation experiments needed compared to  the number of modeling runs.
\end{remark}

We have used McDiarmid's inequality Theorem \ref{conc_mcdiarmid} as
the model for concentration in this paper, but that is not necessary.
All that was needed  is a concentration
parameter ${\cal D}_{F}$ which scales a certain way with 
 sampling. In particular, concentration
theorems that do not require i.i.d.~sampling, for example the  martingale
difference inequality \cite[Thm.~3.14, Page 224]{McDiarmid91}, can be applied to derive results
similar, but more complex, to those obtained.   Another 
 example  of a concentration theorem is the following for 
 Lipschitz
functions, \cite[Cor.~1.17]{Ledoux01}: 
\begin{theorem}
\label{conc_lip}
Let ${\cal X}={\cal X}^{1}\times\cdots \times {\cal X}^{m}$ be the
Cartesian product of
metric
spaces $(X_{i},d_{i})$  with
diameters
$ D_{i},i=1,..,m$ and let ${\cal D}_{{\cal X}}^{2} :=
\sum_{i=1}^{m}{ D_{i}^{2}}$.
Let $F:{\cal X} \rightarrow \R$ be Lipschitz with
respect the $\ell_{1}$ metric $d := \sum_{i=1}^{m}d_{i}$ with
Lipschitz constant  $|F|.$
 Then for any product probability measure $\PR=\mu_{1}
\otimes
\cdots
\otimes \mu_{m}$  we have
\[ \PR(F - \E F \geq r ) \leq e^{-\frac{r^{2}}{2|F|^{2}{\cal
D}_{\cal X}^{2}}}\,
.\]
\end{theorem}
The following, easy to prove, proposition shows that the previous results also apply using the
Lipschitz concentration Theorem \ref{conc_lip}.  
\begin{proposition}
Consider the concentration result and notation of Theorem
\ref{conc_lip}. Then
Theorem \ref{thm_main} and Corollaries 
\ref{cor_val} and \ref{cor_cert}
 hold with ${\cal D}$ replaced by 
$2|F| {\cal D}_{\cal X}$.
\end{proposition}

\section{Extrapolative Validation and Certification}
\label{sec_xtrap}
In this section we consider when we  want to validate a model or certify a physical system
 in a regime where the real physical system is impossible or expensive to sample. That is, suppose
 we wish to validate or certify a random variable
 $\hat{F}(\hat{X})$ which is expensive or impossible to sample but are able 
to sample a related random variable $F(X)$. 
 When samples from $\hat{F}(\hat{X})$ are unavailable we have the following validation result in terms of
the Kolmogorov distance
\begin{equation}
\label{kolmogorov}
d(F,\hat{F}) \Defi \max_{b}{
\bigl|\PR(F \leq b)-\PR(\hat{F} \leq b)\bigr|}
\end{equation}
between two
 random variables $F$ and $\hat{F}$.  The corresponding certification result is very similar, but we omit it
for brevity.

\begin{theorem}[Extrapolative Validation]
\label{thm_xval} 
Let $U=F(X)$ have McDiarmid diameter ${\cal D}_{F}$.
Let $F(X_{i}),i=1,..,n$ be i.i.d. samples of $F(X)$
and define
  $ \langle F\rangle_{n}
 := \frac{1}{n}\sum_{i=1}^{n}F(X_{i})$ to be the sample mean.
Let $0 < p < 1, \,0 < \d_{p} < \min{(p,1-p)}$ and suppose that $\hat{F}:{\hat{\cal X}} \rightarrow \R$ satisfies 
\[ d(F,\hat{F}) \leq \d_{p} \,.\] 
Let $0 < \d_{1},\d_{2} < 1$  and for $0 < t < 1 $
define
 $r_{\cal H} :=  \frac{1}{\sqrt{2}} \sqrt{
\log{(p-\d_{p})^{-1}}} +\frac{1}{\sqrt{2n}} \sqrt{
\log{\d_{1}^{-1}}} $ and
 $r_{\cal K} :=  \frac{1}{\sqrt{2}} \sqrt{
\log{(1-p-\d_{p})^{-1}}} +\frac{1}{\sqrt{2n}} \sqrt{
\log{\d_{2}^{-1}}}. $
Then if  $a-a'\geq  {\cal D}_{F}(r_{\cal H} +r_{\cal K})$  the test of
$
 \{ \PR(\hat{F} \geq a) \geq p\}$ versus
 $ \{ \PR(\hat{F} \geq a') < p\}$
defined by
\[T \Defi
\begin{cases} 1\, ,  &   \langle F\rangle_{n}
 \geq  a- {\cal D}_{F}r_{\cal H}\\              0 \, ,  &
\langle F\rangle_{n} < a- {\cal D}_{F}r_{\cal H} \,
\end{cases}\]
satisfies
\[ \theta_{1} \leq \d_{1}\, ,\]
\[\theta_{2} \leq \d_{2}\, .\]
\end{theorem}
When samples from $\hat{F}$ are available but more expensive than samples from $F$,
 we can use the sample data to estimate the Kolmogorov distance between $F$ and $\hat{F}$  and then
corporate the estimate in the test as discussed in Section
 \ref{sec_est} and afterwords.
The following estimate is efficient in the sense that it uses  the
 concentration of the Kolmogorov-Smirnov statistic of Dvoretzky, Kiefer and Wolfowitz \cite{DvKiWo} improved to have a tight constant by Massart \cite{Massart90}
 (see also \cite[Thm.~12.9]{DeGyLu97}) as follows:
Let $n$ i.i.d~samples be taken from $F$ and let $\PR_{n}$ denote its empirical measure and let $n' \leq n$ i.i.d~samples be taken from $\hat{F}$ and let
 $\PR_{n'}$ denote its empirical measure.
 Then the Dvoretzky, Kiefer and Wolfowitz Theorem states that
\[\PR^{n}\Bigl(\sup_{b \in \R}{\bigl|\PR_{n}(F \leq b)-\PR(F \leq b)\bigr|} > \e\Bigr) \leq 2e^{-2n\e^{2}}\] 
and
\[\PR^{n'}\Bigl(\sup_{b \in \R}{\bigl|\PR_{n'}(\hat{F} \leq b)-\PR(\hat{F} \leq b)\bigr|} > \e\Bigr) \leq 2e^{-2n'\e^{2}}.\]
Let us define 
\[d_{n,n'}(F,\hat{F})\Defi \sup_{b \in \R}{\bigl|\PR_{n}(F \leq b)-\PR_{n'}(\hat{F} \leq b) \bigr|}\]
as an estimator of the Kolmogorov distance $d(F,\hat{F})$ defined in (\ref{kolmogorov}). Then since 
\begin{eqnarray*}
 |d(F,\hat{F})-d_{n,n'}(F,\hat{F})|&=&\Bigl|\sup_{b \in \R}{\bigl|\PR(F \leq b)-\PR(\hat{F} \leq b)\bigr|}-\sup_{b \in \R}{\bigl|\PR_{n}(F \leq b)-\PR_{n'}(\hat{F} \leq b)\bigr|}\Bigr|\\
&\leq & \sup_{b \in \R}{\bigl|\PR_{n}(F \leq b)-\PR(F \leq b)\bigr|}+\sup_{b \in \R}{\bigl|
\PR_{n'}(\hat{F} \leq b)-\PR(\hat{F} \leq b)\bigr|}
\end{eqnarray*}
we use $n'\leq n$ to conclude by a simple union bound that
\begin{eqnarray*}
&&\PR^{n+n'}\bigl(  |d(F,\hat{F})-d_{n,n'}(F,\hat{F})|> \e\bigr)\\
&\leq & \PR^{n}\bigl( \sup_{b \in \R}{\bigl|\PR_{n}(F \leq b)-\PR(F \leq b)\bigr|}> \frac{\e}{2}\bigr)+
\PR^{n'}\bigl(\sup_{b \in 
\R}{\bigl|\PR_{n'}(\hat{F} \leq b)-\PR(\hat{F} \leq b)\bigr|}> \frac{\e}{2}\bigr)\\
&\leq & 4 e^{-\frac{1}{2}n'\e^{2}} \,.
\end{eqnarray*}
That is, we have
\begin{equation*}
 \PR^{n+n'}\bigl(  |d(F,\hat{F})-d_{n,n'}(F,\hat{F})|> \e\bigr)\leq 4 e^{-\frac{1}{2}n'\e^{2}}\, .
\end{equation*}
whose confidence form is
\begin{equation}
 \label{kol_est_conf}  \PR^{n+n'}\Biggl(  \Bigl|d(F,\hat{F})-d_{n,n'}(F,\hat{F})\Bigr|> 
\sqrt{\frac{2\ln{4}+2\ln{\d^{-1}}}{n'}}\Biggr)\leq \d\,  .
\end{equation}
This estimation inequality (\ref{kol_est_conf}) can be used, along the lines of Section \ref{sec_est} and afterword, to prove a version of Theorem 
\ref{thm_xval} where the estimate $d_{n,n'}(F,\hat{F})$ is used instead of the 
Kolmogorov distance $d(F,\hat{F})$.  Moreover, since the test and its performance depend logarithmically on this estimate, we should be able to obtain good tests where
$n'$ is {\em much smaller} than $n$. In particular, we should be able to obtain good tests {\em if} the Kolmogorov distance is small enough- instead of by assuming that it is so. However, for brevity, we do not complete this program here but move to the estimation of diameters in validation and certification tests. 
 
\section{Estimation of Diameters in Hypothesis Tests}
\label{sec_est}
The  validation and certification results, Corollaries \ref{cor_val} and \ref{cor_cert},
require the  value of the diameter ${\cal D}_{F}$  for validation
and ${\cal D}_{F-M}$ and
${\cal D}_{M}$ for certification. In principle the modeling and domain experts should have
much to say about bounding
these values.  However, sample data should also say something about  them.
With the eventual goal of combining expert knowledge about the relevant
diameters
with information from sample data, we now proceed to describe how sample data can be used
to estimate these diameters. This will be accomplished through an estimation procedure and
the introduction of "higher order" concentration parameters.
To that end, we now invert the concentration theorem to
its "confidence version" so that the diameters appear inside the probability statement.
This allows the comparison of the diameter with an estimable parameter and 
a mechanism for incorporating estimates of these parameters in the concentration theorems
and therefore into the definitions of tests and the analysis of their performance.

\subsection{Diameters in  Concentration Theorems}

By a simple function inversion,   McDiarmid's inequality 
can be written
\begin{equation}
\label{mcdiarmid_inv1}
 \PR\Bigl(F - \E F \geq f({\cal D}_{F},\d)\Bigr) \leq \d \,
\end{equation}
where $f(r,\d) :=  \frac{r }{\sqrt{2}} \sqrt{\log{\d^{-1}}}$.
This inversion was used in the
 proof of the main
Theorem  \ref{thm_main}. The following two lemmas reformulate those parts of Theorem
\ref{thm_main} which we will use  as  basic building blocks for developing validation and
certification tests with estimated diameters.
\begin{lemma}
\label{lem_inversion}
Let $0 < p < 1$ and $a,a' \in \R$ and consider the functions
$f_{H}:\R^{3} \rightarrow \R$ and $f_{K}:\R^{3} \rightarrow \R$
defined by
\[ f_{H}(r,r',\d)\Defi
\frac{r}{\sqrt{2}}\sqrt{\log{p^{-1}}}+
\frac{r'}{\sqrt{2}}\sqrt{\log{\d^{-1}}}-a\, ,\]
\[ f_{K}(r,r',\d)\Defi
\frac{r}{\sqrt{2}}\sqrt{\log{(1-p)^{-1}}}+
\frac{r'}{\sqrt{2}}\sqrt{\log{\d^{-1}}}+a'\, .\]
Then for all $0 < \d < 1$ and
 all   $F,F' \in {\cal U}$
which satisfy
 $\E F'=\E F$, we have 
\begin{eqnarray*}
\PR( F' \leq -f_{H}({\cal D}_{F},{\cal D}_{F'},\d)| F \in {\cal H}_{a})& \leq
\d\, , &   \\
\PR( F' \geq f_{K}({\cal D}_{F},{\cal D}_{F'},\d)| F \in {\cal K}_{a'}) & \leq \d\, . &
 \, 
\end{eqnarray*}
\end{lemma}
  The following simple lemma shows how to use the
results of Lemma \ref{lem_inversion} to construct hypothesis tests with
controlled errors. It is formulated in terms of the primary variable $F$, a test
variable $F'$, and a vector $\vec{F}$ of auxiliary variables. 
\begin{lemma}
\label{test_basic}
Let ${\cal H}, {\cal K} \subset {\cal U}$ be null and alternative hypothesis spaces and let
 $ k \in \Bbb{N} $ and $0 <\d_{1},\d_{2}< 1$.
Consider functions $g_{K},g_{H}:{\cal U}^{k}
\rightarrow \R$  such that for  all  
$F,F' \in {\cal U}$ there
exists a vector $\vec{F}$ of auxiliary random variables $F^{j}, j=1,..,k$  such that
\[ \PR\Bigl( F' \leq
-g_{H}(\vec{F})\big| F \in
{\cal
H}\Bigr)\leq \d_{1} ,\]
\[ \PR\Bigl( F' \geq g_{K}(\vec{F})\big| F \in
{\cal
K}\Bigr)
\leq  \d_{2}.\]
We call any such  vector $\vec{F}$  admissible for $F, F'$.
Now  suppose $F, F' \in {\cal U}$ and consider any admissible
and vector $\vec{F}$. Consider the test $T$ of  $ F\in {\cal H}$ against $ F \in {\cal K}$ defined by
\begin{equation}
\label{test1}
T
        \Defi
\begin{cases}
1 \, , & F'  > -g_{H}(\vec{F})\\
0\,  ,&  F' \leq -g_{H}(\vec{F})
\end{cases}
\end{equation}
Then if $g_{H}(\vec{F})+g_{K}(\vec{F}) \leq 0$ we have
\[ \theta_{1}(T) \leq 
\d_{1} ,\]
\[ \theta_{2}(T) \leq 
 \d_{2}.\]
\end{lemma}

In general, concentration theorems can be used to establish results like  Lemma
\ref{lem_inversion} and then Lemma
 \ref{test_basic} can be used to establish a test and bound its errors. In particular, we see
how
the main theorem, Theorem \ref{thm_main}, with test point fixed at the right-hand side of
the interval, can then be obtained by a combined
application of
Lemmas  \ref{lem_inversion} and \ref{test_basic}: first apply Lemma
  \ref{lem_inversion} and then  apply Lemma \ref{test_basic} with ${\cal H}:={\cal H}_{a},{\cal
K}:={\cal H}_{a'},\vec{F}:= (F,F')$
and
\[ g_{H}(\vec{F})=g_{H}(F,F') \Defi f_{H}({\cal D}_{F},{\cal D}_{F'},\d_{1})\,,\]
\[ g_{K}(\vec{F})=g_{K}(F,F') \Defi f_{K}({\cal D}_{F},{\cal D}_{F'},\d_{2})\, .\]
However, what is important here is that we are now in a position define tests which  use estimates of
$ f_{H}({\cal D}_{F},{\cal D}_{F'},\d_{1})$ and $f_{K}({\cal D}_{F},{\cal D}_{F'},\d_{2})$.
Since we see no efficient way of estimating the McDiarmid diameter ${\cal D}_{F}$ of a
function $F$ but we do know something about the estimation of the usual diameter
  $ D_{G}$ of a function $G$
defined by
\[  D_{G}\Defi \sup_{x,x' \in {\cal X}
}{\bigl(G(x)-G(x')\bigr)},\]
we ask whether we can estimate the McDiarmid diameter by estimating the usual
diameters of a
set of auxiliary set of functions. 
To that end we first introduce
a relationship between the
McDiarmid diameter and the usual  diameters of a set of
auxiliary observables. These latter diameters we will then estimate using extreme value
estimators 
 in Section \ref{sec_conc}.
Now, ignoring for the moment the question of the attainment of suprema, if we  define
\[F^{j}(x_{j})
 \Defi
F(x_{1}^{*},..,x^{*}_{j-1},x_{j},x^{*}_{j+1},..,x^{*}_{m})\] 
where
\begin{eqnarray*}
&& (x_{1}^{*},..,x^{*}_{j-1},x^{*}_{j+1},..,x^{*}_{m})\\
& \Defi& 
\arg \max_{x_{1},..,x_{j-1},x_{j+1},..,x_{m}}\,\max_{x_{j}, x'_{j}} 
 {\Bigl(F(x_{1},..,x_{j-1},x_{j},x_{j+1},..,x_{m})-F(x_{1},..,x_{j-1},x'_{j},
x_{j+1},..,x_{m})\Bigr)}
\end{eqnarray*}
 it follows that
\[ {\cal D}^{2}_{F} = \sum_{j=1}^{m}{D^{2}_{F^{j}}}\, .\]
Namely  the McDiarmid diameter is a function of diameters.
However, this relation will only be of use to us if the
 functions $F^{j},j=1,..,k$ are observable, namely, they can be
evaluated. Now suppose we are in possession of a set $F^{j},j=1,..,k$ auxiliary observables
and let $\vec{D}$ denote the vector of their diameters. Suppose we also have
 functions $g_{H}$ and $g_{K}$ such that
\[  f_{H}({\cal D}_{F},{\cal D}_{F'},\d) \leq  g_{H}(\vec{D},\d)\,, \quad   0 < \d <1,\]
\[  f_{K}({\cal D}_{F},{\cal D}_{F'},\d) \leq
g_{K}(\vec{D},\d)\,, \quad  0 < \d <1 .\] 
Then since Lemma  \ref{lem_inversion}  asserts that for all $0 < \d < 1$ we have
\begin{eqnarray*}
\PR^{n}\Bigl( F' \leq -f_{H}\bigl(g(D),g'(D),\d\bigr)\big| {\cal H}\Bigr)& \leq
\d\, , &   \\
\PR^{n}\Bigl( F' \geq f_{K}\bigl(g(D),g'(D),\d\bigr)\big| {\cal K}\Bigr) & \leq \d , &
\end{eqnarray*}
it follows easily that for all $0 < \d < 1$  we have
\begin{eqnarray}
\label{eq_main}
\PR^{n}\Bigl( F' \leq -g_{H}\bigl(\vec{D},\d\bigr)\big| {\cal H}\Bigr)& \leq
\d\, , &   \\
\PR^{n}\Bigl( F' \geq g_{K}\bigl(\vec{D},\d\bigr)\big| {\cal K}\Bigr) & \leq \d , &
\end{eqnarray}
Consequently, we can apply Lemma \ref{test_basic} to obtain tests defined instead in terms of the
estimable functions $g_{H}\bigl(\vec{D},\d\bigr)$ and   $g_{K}\bigl(\vec{D},\d\bigr)$.
Most importantly, {\em the inequalities (\ref{eq_main}) 
 remain valid with the vector of diameters $\vec{D}$ replaced by the
vector of essential
diameters}.
 When the essential
diameter is {\em much} smaller than the given diameter,  this difference
can often offset the looseness
corresponding  to the error associated with estimating the essential
diameter using the empirical diameter.

Let us now give the first important example of  auxiliary observables. In
this case, they will be none other than the functions $F, F'$ themselves, but will
require the introduction of new functions, $c_{F},c_{F'}$ of $F$ and $F'$ which will have to be
approximately known.
To that end, define a coefficient of the separability $c_{F}$ of the function $F$
 with respect to the $m$ components of $ {\cal X}:=
\prod_{j=1}^{m}{\cal X}^{j}$ as follows:
\begin{definition}
\label{cf}
Let  ${\cal X}
:=
\prod_{j=1}^{m}{\cal X}^{j}$ be a product and consider a
function $F:{\cal X}
\rightarrow \R$, its diameter $ D_{F}$,
 and its McDiarmid diameter ${\cal D}_{F}.$
We define the coefficient of separability $c_{F}$ with respect to the
product
${\cal X}$
 to be
\[ c_{F} \Defi  \frac{  {\cal D}_{F}}{ D_{F}}.\]
\end{definition}
With this definition it is clear that if we define
$g_{H}(D_{F},D_{F'},\d):= f_{H}(c_{F}D_{F}, c_{F'}D_{F'},\d)$
and $g_{K}(D_{F},D_{F'},\d):= f_{K}(c_{F}D_{F}, c_{F'}D_{F'},\d)$, where we suppress the
dependency on $c_{F},c_{F'}$, we have
\[  f_{H}({\cal D}_{F},{\cal D}_{F'},\d) =  g_{H}(D_{F}, D_{F'},\d)\,, \quad   0 < \d <1,\]
\[  f_{K}({\cal D}_{F},{\cal D}_{F'},\d) =
g_{K}(D_{F},D_{F'},\d)\,, \quad  0 < \d <1 \] and therefore
\begin{eqnarray*}
\PR^{n}\Bigl( F' \leq -g_{H}\bigl(D_{F},D_{F'},\d\bigr)\big| {\cal H}\Bigr)& \leq
\d\, , &   \\
\PR^{n}\Bigl( F' \geq g_{K}\bigl(D_{F},D_{F'},\d\bigr)\big| {\cal K}\Bigr) & \leq \d , &
\end{eqnarray*}

Although we have now introduced a new function $c_{F}$ which will have to be known or well
bounded, this function has nice properties, which we now describe, which make assuming its value a weaker assumption
than assuming the value of a McDiarmid diameter.
Let us say that a map $\phi: \prod_{j=1}^{m}{\cal
X}^{j}\rightarrow
\prod_{j=1}^{m}{\cal
X}'^{j} $ is a diagonal bijection if it is a
product map  $\phi =\prod_{j=1}^{m}{\phi}_{j}$
such that ${\phi}_{j}:{\cal X}^{j}\rightarrow{\cal X}'^{j}$ is  a
bijection for all $j=1,..,m$. The following lemma shows that $ F \mapsto
c_{F}$
is a bounded
  invariant
under non-singular affine transformations $F \mapsto aF+b$ of the
function
$F$ and a
diagonal bijective invariant. 
\begin{lemma}
\label{c_basic}
The mapping $F \mapsto c_{F}$ is a diagonal bijective invariant. Moreover,
we have
\[c_{aF+b}=c_{F},\quad a, b \in \R, a \neq 0 \]
and
\[\frac{1}{\sqrt{m}} \leq c_{F} \leq \sqrt{m}.\]
\end{lemma}
In Example \ref{ex_c} in the Appendix we describe the  attainment of the
the extreme
case $c_{F}=\frac{1}{\sqrt{m}}$ and  $c_{F}=\sqrt{m}$:  roughly, the
lower bound
is attained for functions which are separable in the $m$
components and the upper bound is obtained for a function related to the
Euclidean metric\footnote{
In personal communication, L. Gurvits has  demonstrated 
that  nontrivial lower bounds may not exist when ${\cal X}$ is not a product.
 For example it is easy to construct cases where the
 partial diameters are all zero and the
diameter is not.}.

Lemma \ref{test_basic} states that conditions such as
\begin{eqnarray}
\label{eq_main1}
\PR^{n}\Bigl( F' \leq -g_{H}(D)\big| {\cal
H}\Bigr)& \leq
\d_{1}\, , &   \\
\PR^{n}\Bigl( F' \geq g_{K}(D)\big| {\cal K}\Bigr) &
\leq \d_{2} . &
\end{eqnarray}
and
 $g_{H}(D)+g_{K}(D) \leq 0$ for functions of the {\em
essential} diameter
vector of auxiliary observables 
are sufficient to develop a good test. In the above analysis this was
accomplished by knowledge about the coefficients of separability $c_{F},c_{F'}$ which
allowed us to use the functions $F$ and $F'$ as their own auxiliary observables.
However relations such as (\ref{eq_main1}), where the inequalities are in terms of the usual
diameters,  can be obtained through other
concentration inequalities. For example, if we instead appeal to  the  Lipschitz concentration
  Theorem
\ref{conc_lip},
it is easy to obtain
inequalities (\ref{eq_main1}) from  Lemma \ref{lem_inversion} 
with ${\cal D}_{F}$ replaced
  by $2|F|{\cal D}_{X}$.   However, it is easy to show that
\[ {\cal D}_{F}\leq |F|{\cal D}_{{\cal X}}\]
indicating that   McDiarmid's Theorem
\ref{conc_mcdiarmid}
 provides a superior
concentration guarantee. On the other hand, since ${\cal D}_{{\cal X}}$ is a sum of
diameters, and
the supremum \[|F| \Defi \sup_{x\neq x'}{\frac{F(x)-F(x')}{d(x,x')}}\]
can be estimated by the empirical Lipschitz coefficient
\[ \hat{|F|} \Defi \sup_{X_{i}\neq
X_{i'},i,i'=1,..,n}{\frac{F(X_{i})-F(X_{i'})}{d(X_{i},X_{i'})}},\]
it follows that
  $|F|$ and
${\cal D}_{{\cal X}}$ can
be estimated from sample data using Corollary \ref{cor_diam2}  in Section 
\ref{sec_conc}.
However, this will require
 the random variable
$X$ to be  observable.
 Consequently, when ${\cal D}_{F}$ has no
readily apparent auxiliary observables (such as when no knowledge of
the coefficient of separability $c_{F}$ is available), and $X$ is observable,
 using  the  Lipschitz concentration Theorem \ref{conc_lip} may prove fruitful. 

\subsection{Concentration of Empirical Quantiles}
\label{sec_conc}
Since we will be concerned with the effects of estimating essential
diameters using sample data, we now
 describe results, of independent interest, concerning the concentration of
empirical quantiles about distributional quantiles and show how to use them to bound
empirical diameters with respect to essential diameters.
Let $X$ be a real random variable with probability measure $\PR$
and recall its distribution function
$\F(\xi) := \PR(X \leq \xi)$.
 For $0 < p < 1$ define the quantiles
$ \xi_{p}:= \F^{-1}(p) := \inf{\{\xi: \F(\xi) \geq p}\}$. We will
use
important properties of $\F$ and $\F^{-1}$ listed in Theorem \ref{lem_F}
in the
Appendix.
Moreover, let $X_{i},i=1,n$
be i.i.d samples from $X$. Let $\PR_{n}$ denote the corresponding
empirical
measure, denote by $\F_{n}$ its corresponding distribution function,
and
let $\hat{\xi}_{p}:=  \inf{\{\xi: \F_{n}(\xi) \geq p\}} $ denote the
empirical quantiles.
We will use the following improvement of a theorem of Serfling
\cite[Thm.~2.3.2]{serfling}.
\begin{theorem}
\label{thm_serfling}
Let $0 < q < 1 $ and suppose that $\xi >\xi_{q}.$ Then with $\d_{1} := \F(\xi)-q$ we have
\begin{enumerate}
\item
$\PR^{n}(\hat{\xi}_{q}>\xi) \leq e^{-2n\d_{1}^{2}}$
\item
$\PR^{n}(\hat{\xi}_{q}>\xi) \leq
e^{-\frac{n\d_{1}^{2}}{2(1-\F(\xi))+\frac{2}{3}\d_{1}}}$
\item
$\PR^{n}(\hat{\xi}_{q}>\xi) \leq e^{-\frac{n\d_{1}^{2}}{2\F(\xi)}}$
\end{enumerate}
On the other hand suppose that $\xi <\xi_{q}.$ Then with $\d_{2} :=
q-
\F(\xi)$ we have
\begin{enumerate}
\item
$\PR^{n}(\hat{\xi}_{q}<\xi) \leq e^{-2n\d_{2}^{2}}$
\item
$\PR^{n}(\hat{\xi}_{q}<\xi) \leq
e^{-\frac{n\d_{2}^{2}}{2\F(\xi)+\frac{2}{3}\d_{2}}}$
\item
$\PR^{n}(\hat{\xi}_{q}<\xi) \leq e^{-\frac{n\d_{2}^{2}}{2(1-\F(\xi))}}$
\end{enumerate}
\end{theorem}
Theorem \ref{thm_serfling} now gives us a good tool to compare empirical
diameters with quantiles.
\begin{theorem}
\label{thm_serfling2}
Let $X$ be  a real random variable and let $X_{i},i=1,..,n$ i.i.d.~samples.
For $0 < p < 1$,  we have
\begin{enumerate}
\item
  Let
$D_{n}:= \sup_{i=1,..,n}{X_{i}}-\inf_{i=1,..,n}{X_{i}}$
denote the empirical range.
 Then  we have
\[\PR^{n}\Bigl(
 D_{n}<\xi_{p}-\xi_{1-p}\Bigr)
\leq
2e^{-\frac{1}{2}n(1-p)}\, .
\]
\item
Suppose $X$ is a non-negative random variable and let $S_{n} :=
\sup_{i=1,..,n}{X_{i}}$ denote
the empirical supremum.  Then we have
\[ \PR^{n}\Bigl(S_{n} <\xi_{p}\Bigr) \leq
e^{-\frac{1}{2}n(1-p)}\, .
\]
\end{enumerate}
\end{theorem}
We now show how to use
Theorem \ref{thm_serfling2} to
bound the the empirical diameters in terms of essential diameters.
To that end, let
$X_{-} := ess\inf X$ and $X_{+}:= ess\sup X$. Then the
  essential diameter
is $ D := X_{+}-X_{-}$. We introduce a {\em tail function} quantifying the behavior of a random variable near its range limit.
\begin{definition}
\label{def_tail}
Let the tail function $\t^{X}$ corresponding to $X$ be defined by
\begin{eqnarray}
\label{tail}
\t^{X}(\epsilon) \Defi & \sup{t}& \quad \epsilon > 0\\
&t: \x_{1-t}-\xi_{t} \geq \frac{ D}{1+\epsilon},
\end{eqnarray}
\end{definition}
Roughly speaking the function
 $\t^{X}(\epsilon)$ is such that
 the set obtained  by eliminating the right and left tails
 of mass
$\t$ is at least $\frac{1}{1+\epsilon}$ as large as  the diameter.
Characterization
of tail behaviors lies as the heart of the theory of the limiting behavior of extreme
order statistics (see e.g.~Arov and Bobrov\cite{ArBo60},
Pickands \cite{Pickands75},
Barndorff-Neilsen \cite{Barndorff63}) and will no doubt be useful when
the diameters are unbounded,
but since we concern ourselves with the bounded case here,
 the tail function (\ref{tail})
appear sufficient to our needs.
The following proposition provides a lower bound for $\t(\epsilon)$ in terms of the
distribution function for $X$.
\begin{proposition}
\label{d_suff}
Let $X$ be a real random variable and
suppose that $X_{-} := ess\inf X$ and $X_{+}:= ess\sup X$ are finite.
Then in terms of the essential
diameter
 $ D := X_{+}-X_{-}$, we
 have \[\t^{X}(\e)  \geq \min{\Bigl(\F(X_{-}+ \frac{\e}{2(1+\e)}  D),
1-\F(X_{+}-\frac{\e}{2(1+\e)} D)\Bigr)}, \quad \e > 0.\]
\end{proposition}
As an elementary application, consider the case where
 the tails are not too thin. That is suppose
for some
$\kappa
>0$ we have $\F(X_{+}-x)
<1-\bigl(\frac{x}{ D}\bigr)^{\kappa}, 0 < x <  D$ and $\F(X_{-}+x)
\geq  \bigl(\frac{x}{ D}\bigr)^{\kappa}, 0 < x <  D.$
We conclude from Proposition  \ref{d_suff} that
\[\t^{X}(\epsilon) \geq \Bigl(\frac{\epsilon}{2(1+\epsilon)}\Bigr)^{\kappa}
\geq \Bigl(\frac{\epsilon}{4}\Bigr)^{\kappa} \] which for $\epsilon $ small is
$t^{X}(\epsilon) \gtrapprox \bigl(\frac{\epsilon}{2}\bigr)^{\kappa}\,. $
For non-negative random variables we proceed similarly to definition
(\ref{tail}) and define
\begin{eqnarray}
\label{tail_d}
\t^{X}_{+}(\epsilon) \Defi & \sup{\d}&\\
&\xi_{1-\d} \geq \frac{X_{+}}{1+\epsilon}& 
\end{eqnarray}
Similar arguments used in the proof of Proposition $\ref{d_suff}$ imply that
 $\t^{X}_{+}(\epsilon) \geq 1-\F( \frac{X_{+}}{1+\epsilon}).$
We are now in a position to compare the empirical diameter with the essential
diameter using the tail function $\t$.
\begin{corollary}
\label{cor_diam2}
Let  $X$ be a real random variable and  let $X_{i},i=1,..,n$ be
i.i.d.~samples from $X$.
Then
\begin{enumerate}
\item
Let
$ D_{n}:= \sup_{i=1,..,n}{X_{i}}-\inf_{i=1,..,n}{X_{i}}$
denote the empirical range and let $\t$ be define by
(\ref{tail}).
Then for all $\e > 0$ 
  we have
\[ \PR^{n}\Bigl( D_{n} <\frac{ D}{1+\epsilon}\Bigr) \leq
2e^{-\frac{n\t(\epsilon) }{2}}\, .
\]
\item
Suppose $X$ is a non-negative random variable and let $S_{n}
:=
\sup_{i=1,..,n}{X_{i}}$ denote
the empirical supremum and let  $\t_{+}$ be define by
(\ref{tail_d}).
Then for all $\e > 0$
 we have
\[ \PR^{n}\Bigl(S_{n} <\frac{X_{+}}{1+\epsilon}\Bigr) \leq
e^{-\frac{n\t_{+}(\epsilon)}{2}}\, .
\]
\end{enumerate}
\end{corollary}

\subsection{Estimation in  Hypothesis Tests}
Lemma \ref{test_basic} and the discussion thereafter shows that when
 $f_{H}(D)+f_{K}(D) \leq 0$ (thus determining a relationship between the performance
thresholds $a,a'$ and the diameter $D$) the test of Lemma \ref{test_basic}
 of ${\cal H}$
 against ${\cal K}$ has type $I$ error not greater than $\d_{1}$ and
type II error not greater than $\d_{2}$.
 However, when good upper bounds on $D$ are not known and thus it is not
known if  $f_{H}(D)+f_{K}(D) \leq 0$,  these results may
be of limited value. To resolve this situation we use
sample data to estimate $D$ and use the estimate to test the condition
 $f_{H}(D)+f_{K}(D) \leq 0$.
To develop validation and certification tests along the lines
above will involve sequential tests. The type of test we consider we
call a 
 stop option hypothesis test:
\begin{definition}
\label{stoptest}
For $i=1,2$ consider a null hypothesis ${\cal H}_{i}$ and alternative ${\cal
K}_{i}$ of sets of real random variables,    
and a test $T_{i}$ of  ${\cal H}_{i}$ against ${\cal K}_{i}$.  
Define the reduced hypothesis spaces
\[ {\cal H}_{2\epsilon} \Defi ({\cal K}_{1} \cup {\cal
H}_{1})\cap{\cal H}_{2},\]
\[ {\cal K}_{2\epsilon} \Defi ({\cal K}_{1} \cup {\cal
H}_{1})\cap{\cal K}_{2}.\]
 We define the stop
option test $T_{1} \blacktriangleleft T_{2}$ which first implements $T_{1}$ and if
the outcome is acceptance, to use $T_{2}$ to test ${\cal H}_{2\epsilon}$ against
${\cal K}_{2\epsilon}$:
\[ T_{1} \blacktriangleleft T_{2}
\Defi 
\begin{cases}
0& T_{1}=0\\
(1,0) & T_{1}=1, T_{2}=0\\
(1,1) & T_{1}=1, T_{2}=1
\end{cases}
\]
All types of errors for the test $T_{1} \blacktriangleleft T_{2}$ can be controlled
by the following three  
 types of errors:
\[\theta_{1}(T_{1} \blacktriangleleft T_{2})\Defi \PR\bigl(T_{1}=0\big|{\cal H}_{1}\bigr)\]
\[\theta_{11}(T_{1} \blacktriangleleft T_{2})\Defi \PR\bigl(\{T_{1}=1, 
T_{2}=0\} \big|{\cal K}_{1} \cup ({\cal H}_{1}\cap{\cal H}_{2})\bigr)\]
\[\theta_{12}(T_{1} \blacktriangleleft T_{2})\Defi \PR\bigl(\{T_{1}=1,
T_{2}=1\} \big|{\cal K}_{1} \cup ({\cal H}_{1}\cap{\cal K}_{2})\bigr)\]

\end{definition}

Since \[{\cal K}_{1} \cup ({\cal H}_{1}\cap{\cal H}_{2})=({\cal K}_{1} \cup
{\cal
H}_{1})\cap{\cal H}_{2}= {\cal H}_{2\epsilon}\,  ,\]
\[{\cal K}_{1} \cup ({\cal H}_{1}\cap{\cal K}_{2})=({\cal K}_{1} \cup {\cal
H}_{1})\cap{\cal K}_{2}={\cal K}_{2\epsilon}\, , \]
it follows that 
\[ \theta_{11}(T_{1} \blacktriangleleft T_{2})=\PR\bigl(\{T_{1}=1,
T_{2}=0\} \big|{\cal H}_{2\epsilon})\bigr)\, ,\]
\[ \theta_{12}(T_{1} \blacktriangleleft T_{2})=\PR\bigl(\{T_{1}=1,
T_{2}=1\} \big|{\cal K}_{2\epsilon})\bigr)\, .\]
Consequently, the stop option test $T_{1} \blacktriangleleft T_{2}$
converts tests of ${\cal H}_{i}$ against ${\cal K }_{i},i=1,2$ into a
test of ${\cal H}_{1}$ against ${\cal K }_{1}$ and if accepted then
tests ${\cal H}_{2\epsilon}$ against ${\cal K }_{2\epsilon}$.
Since 
\begin{eqnarray*}
 \PR\bigl(T_{1}=1
\big|{\cal K}_{1} \bigr)&=&\PR\bigl(\{T_{1}=1,T_{2}=0\}
\big|{\cal K}_{1} \bigr)+\PR\bigl(\{T_{1}=1,T_{2}=1\}
\big|{\cal K}_{1} \bigr)\\
& \leq&  \PR\bigl(\{T_{1}=1,
T_{2}=0\} \big|{\cal K}_{1} \cup ({\cal H}_{1}\cap{\cal H}_{2})\bigr)+
\PR\bigl(\{T_{1}=1,
T_{2}=1\} \big|{\cal K}_{1} \cup ({\cal H}_{1}\cap{\cal K}_{2})\bigr)\\
&=& \theta_{11}+\theta_{12},
\end{eqnarray*} 
it follows that if all the errors $\theta_{1},\theta_{11},\theta_{12},$ are
small, then given ${\cal K}_{1}$ with
high probability we reject ${\cal H}_{1}$ and stop, and given ${\cal H}_{1}$ with high probability
we accept ${\cal H}_{1}$ and test well on the second test $T_{2}$ when applied to
the reduced null hypothesis ${\cal H}_{2\epsilon}$
against the reduced alternative ${\cal K}_{2\epsilon}.$
Finally, we note that we can also define the errors to be conditional errors
as in the conditional hypothesis testing framework analyzed in
\cite{Kiefer76}. In the applications of this paper, one can show that given 
${\cal H}_{1}$ the conditional errors are roughly the same as above,
and given ${\cal K}_{1}$, the conditional errors are not good. However,
in this case with high probability the first test will reject and stop.


We now proceed to implement the stop option test in the validation and certification
setting. To simplify the analysis
 in the following theorem, instead of first
testing  ${\cal H}_{1}=\{f_{H}(D)+f_{K}(D)
\leq 0\}$ against
 ${\cal K}_{1} =\{f_{H}(D)+f_{K}(D) > 0\}$, (where $D$ is the essential
diameter vector) we test
${\cal H}_{1}=\{f_{H}((1+\epsilon)D)+f_{K}((1+\epsilon)D)
\leq 0\}$ against
 ${\cal K}_{1} =\{f_{H}(D)+f_{K}(D) > 0\}.$ Also observe that this result is stated in terms of auxiliary variables which are sampled
 concomitantly with the sampling of the primary variable $F$. More general situations can be easily addressed.  

\begin{theorem}
\label{thm_test}
Let 
${\cal H}_{2}$ and 
${\cal K}_{2}$ denote null and alternate hypothesis spaces of real random
variables.
Let $X$ be a random variable with range ${\cal X}$
 and probability law $\PR$, and let
$F:{\cal X} \rightarrow \R$ and $F':{\cal X}^{n} \rightarrow \R$. Consider non-observable i.i.d.~samples $X_{i},i=1,..,n$ and observable $F'(
X_{1},..,X_{n})
$. In addition, let $k$ be a positive integer and
 let
 $F^{j}:{\cal X} \rightarrow \R,j=1,..,k$ be a collection of auxiliary observables with essential diameters
 $D_{F^{j}},j=1,..,k$. 
Let $ D:= \langle D_{F^{j}}\rangle_{j=1,..,k}$
denote the corresponding vector of essential diameters, 
\[   \hat{D}_{F^{j}} \Defi \sup_{i=1,..,n}
F^{j}(X_{i})
-\inf_{i=1,..,n} F^{j}(X_{i}) \]
denote the empirical diameters, and
 $ \hat{ D}:= \langle\hat{D}_{F^{j}}\rangle_{j=1,..,k}$
 the vector of empirical diameters.
  Let $f_{H}:\R^{k} \rightarrow \R$ and
$f_{K}:\R^{k} \rightarrow \R$ be non-decreasing functions such that
\begin{eqnarray*}
\PR^{n}( F' \leq -f_{H}(D)| {\cal H}_{2})& \leq & \d_{1},\\
 \PR^{n}( F' \geq f_{K}(D)| {\cal K}_{2}) & \leq & \d_{2}.
\end{eqnarray*}
 Let $\e > 0$ and define
 \begin{eqnarray*}
{\cal H}_{1}&=&\{f_{H}((1+\epsilon)D)+f_{K}((1+\epsilon)D)
\leq 0\},\\
{\cal K}_{1}& =&\{f_{H}(D)+f_{K}(D) > 0\}.
\end{eqnarray*}
and define  the test $T_{1}$ of  ${\cal H}_{1}$ against the alternative ${\cal
K}_{1}$ by
\begin{equation}
\label{teste}
T_{1}
\Defi
\begin{cases}
  1 \, , &
f_{H}((1+\epsilon)\hat{D})+f_{K}((1+\epsilon)\hat{D}) \leq 0\\
  0\,  ,&
f_{H}((1+\epsilon)\hat{D})+f_{K}((1+\epsilon)\hat{D}) > 0\, .
\end{cases}
\end{equation}
Moreover, define the
test $T_{2}$ of ${\cal H}_{2}$ against ${\cal K}_{2}$  by
\begin{equation}
\label{test2}
T_{2}
\Defi
\begin{cases}
  1 \, , & F'(X_{1},..,X_{n})  > -f_{H}((1+\epsilon)\hat{D})\\
  0\,  ,&  F'(X_{1},..,X_{n}) \leq -f_{H}((1+\epsilon)\hat{D})\, .
\end{cases}
\end{equation}
Finally, consider the stop option test $T_{1} \blacktriangleleft T_{2}$ and it associated
errors $\theta_{1},\theta_{11}$, and $\theta_{12}$, defined in
Definition \ref{stoptest}. 
Then  if we define
\[ \Delta \Defi   \sum_{j=1}^{k}\PR^{n}\Bigl( (1+\epsilon)\hat{
D}_{F^{j}} < D_{F^{j}}\Bigr),
\]
we have
\begin{eqnarray*}
\theta_{1}&= & 0 \\
\theta_{12}& \leq  &  \d_{1}+ \Delta \\
 \theta_{12} & \leq &  \d_{2}+ \Delta \, .
\end{eqnarray*}
Moreover, for $j=1,..,k$,
 let $\t^{j}$ denote the tail
functions of  $F^{j}$ defined in
(\ref{tail}), and  let 
\begin{equation}
\label{n_def}
 n_{j}^{\epsilon}(\d) \Defi
\frac{2\log{\frac{2k}{\d}}}{\t^{j}(\epsilon)}\, .
\end{equation}
Then  if  $  n \geq \max{\bigl( n_{j}^{\epsilon}(\d_{1}), 
n_{j}^{\epsilon}(\d_{2})\bigr)}, j=1,..,k $  we have
\begin{eqnarray*}
\theta_{1}&= & 0 \\
\theta_{11}& \leq  & 2\d_{1}\\
 \theta_{12} & \leq &  2\d_{2}\, .
\end{eqnarray*}

\end{theorem}

The constraint $ n \geq n^{\epsilon}(\d)$ is logarithmic in $\d^{-1}$ with
multiplier $\frac{2}{\t(\epsilon)}$. If $\t(\epsilon)$ is not too small
 then this is a weak constraint. For fixed sample sizes, this
relation can be used to determine a lower bound on the size of
$\epsilon$ which can be used.

\section{Validation and Certification with Estimated Diameters}
We now present tests for validation and certification using estimated diameters. They show that if the coefficients of
separability $c$ are approximately known
 then the validation and
certification Corollaries \ref{cor_val} and \ref{cor_cert}, using the estimated
diameters, are still good.
\begin{corollary}[Validation with Estimated Diameters]
\label{cor_val_e} 
Let $a' \leq a $, $ 0 <  p < 1$ and let
\[{\cal H}_{2} = \{ \PR(F \geq a) \geq p\}\]
\[{\cal K}_{2} = \{ \PR(F \geq a') < p\}.\]
With the assumptions of Theorem \ref{thm_test}, 
 let
\[F'(X_{1},..,X_{n})
 \Defi \langle F \rangle_{n}\] be the sample mean.
Let $c \geq c_{F}$ be a known constant and consider   
 $F:{\cal X}\rightarrow \R$ with essential diameter $D$ as the 
 only auxiliary observable.
Let  $\t$ be the tail function
of $F$ defined in (\ref{tail}). Let $\e > 0$  and let  $n^{\epsilon}(\d)$ be defined in
(\ref{n_def}) with $k=1$.
Let $0 <
\d_{1},\d_{2} < 1$  
 and  define
\[ f_{H}(r)\Defi  \frac{c r}{\sqrt{2}}\sqrt{\log p^{-1}} +\frac{c
r}{\sqrt{2n}}\sqrt{\log
\d_{1}^{-1}}-a\]
\[f_{K}(r)\Defi  \frac{c r}{\sqrt{2}}\sqrt{\log(1-p)^{-1}} +\frac{c
r}{\sqrt{2n}}\sqrt{\log
\d_{2}^{-1}}+a' \, \] and let
\[   \hat{D}\Defi  \sup_{i=1,..,n}
F(X_{i})
-\inf_{i=1,..,n} F(X_{i})\] denote the empirical diameter.
Moreover,  consider the
 stop option test $T_{1} \blacktriangleleft T_{2}$ and its
associated
errors $\theta_{1},\theta_{11}$, and $\theta_{12}$, as  in
Theorem \ref{thm_test}. Then for $n \geq \max{( n^{\epsilon}(\d_{1}),
n^{\epsilon}(\d_{2}))}$ we have
\begin{eqnarray*}
\theta_{1}&= & 0 \\
\theta_{11}& \leq  & 2\d_{1}\\
 \theta_{12} & \leq &  2\d_{2}\, .
\end{eqnarray*}
  
\end{corollary}

For certification,
we now address the case where $F=F_{1}+F_{2}$, where $F_{1}$ can
represent the model and $F_{2}$ represent the difference between the
physical system and the model.
\begin{corollary}[Certification with Estimated Diameters]
\label{cor_cert_e}
Let $a' \leq a $, $ 0 <  p < 1$ and let
\[{\cal H}_{2} = \{ \PR(F \geq a) \geq p\}\]
\[{\cal K}_{2} = \{ \PR(F \geq a') < p\}.\]
With the assumptions of Theorem \ref{thm_test},
 let $F=F_{1}+F_{2}$. Let $n=n_{1} +n_{2}$ and in terms of the observation
$F_{1}(X_{i}), i=1,..,n_{1}$ and $F_{2}(X_{i}), i=n_{1}+1, n_{1}+n_{2}$ define
\[F'(X_{1},..,X_{n})
 \Defi
\langle F_{1}\rangle_{n_{1}}+
\langle F_{2}\rangle_{n_{2}}.\]
Let  $c_{1} \geq c_{F_{1}}$ and $c_{2} \geq c_{F_{2}}$ be known
constants and consider
 $F_{1}$ with essential diameter $D_{1}$ and $F_{2}$ with essential
diameter $D_{2}$ as 
  auxiliary observables, with diameter vector $D:= (D_{1},D_{2})$.
  Let  $\t^{j},j=1,2$ be the tail
functions, defined in (\ref{tail}),
of $F_{1}$ and $F_{2}$ respectively. Let   $\e > 0$ and  
  let  $n_{j}^{\epsilon}(\d),j=1,2$ be defined in
(\ref{n_def}) with $k=2$.
Let $\d_{1},\d_{2} < 1$
 and  define
\[ f_{H}(s_{1},s_{2})\Defi  (c_{1}s_{1}+c_{2}s_{2})\sqrt{\log p^{-1}} +
\sqrt{\Bigl(\frac{c^{2}_{1}s^{2}_{1}}{2n_{1}}
+\frac{c_{2}^{2}
s^{2}_{2}}{2n_{2}}\Bigr) 
\log \d_{1}^{-1}}-a\]
\[f_{K}(s_{1},s_{2})\Defi  (c_{1}s_{1}+c_{2}s_{2})\sqrt{\log (1-p)^{-1}} +
\sqrt{\Bigl(\frac{c^{2}_{1}s^{2}_{1}}{2n_{1}}
+\frac{c_{2}^{2}
s^{2}_{2}}{2n_{2}}\Bigr)
\log \d_{2}^{-1}}
+a' \, \] and let
\[   \hat{D_{1}}\Defi  \sup_{i=1,..,n_{1}}
F_{1}(X_{i})
-\inf_{i=1,..,n_{1}} F_{2}(X_{i})\]
\[   \hat{D_{2}}\Defi  \sup_{i=n_{1}+1,..,n_{1}+n_{2}}
F_{2}(X_{i})
-\inf_{i=n_{1}+1,..,n_{1}+n_{2}} F_{2}(X_{i})\]
 denote the empirical diameters with
empirical diameter vector $\hat{D}:= (\hat{D_{1}},\hat{D_{2}})$.
Moreover,  consider the
 stop option test $T_{1} \blacktriangleleft T_{2}$ and its
associated
errors $\theta_{1},\theta_{11}$, and $\theta_{12}$, as  in
Theorem \ref{thm_test}. Then for $n_{j} \geq \max{( n_{j}^{\epsilon}(\d_{1}),
n_{j}^{\epsilon}(\d_{2}))},j=1,2$ we have
\begin{eqnarray*}
\theta_{1}&= & 0 \\
\theta_{11}& \leq  & 2\d_{1}\\
 \theta_{12} & \leq &  2\d_{2}\, .
\end{eqnarray*}

\end{corollary}

\section{Proofs}

\begin{proofof}{Theorem \ref{thm_main}}

To begin we first prove the following simple lemma
 that quantifies how the mass constraints of the null  
 ${\cal H}^{{\cal D}}_{a}$ or the alternative
${\cal K}^{{\cal D}}_{a}$ imply a constraint on the value of $\E U$.

\begin{lemma}
\label{lem_mean}
Let $ 0 < p <1$, $a\in \R$ and ${\cal D}> 0$.
For $0 < t < 1$ 
  let $r_{t}  := \frac{{\cal D}}{\sqrt{2}} \sqrt{   \log{t^{-1}}}.$ 
Suppose $U \in {\cal H}^{{\cal D}}_{a}$, then $\E U \geq a-r_{p}$.
Suppose $U \in {\cal K}^{{\cal D}}_{a}$, then $\E U \leq a+r_{1-p}$.
\end{lemma}

\begin{proof}
Let $U \in {\cal H}^{D}_{a}$ and suppose to the contrary that $\E U
<a-r_{p}$. Then we have
\[ p \leq \PR(U \geq a) \leq \PR(U > \E U +r_{p}) =
 \PR(U -  \E U >
r_{p})  < p \]
which is a contradiction, thus establishing the first assertion.
Now let $U \in {\cal K}^{D}_{a}$ and suppose to the contrary that
$ \E U  >
a+r_{1-p}$. Then
\[ p >  \PR(U \geq a) \geq \PR(U > \E U -r_{1-p}) =  \PR(U
-  \E U > -
r_{1-p})=1-\PR(U
-  \E U \leq -
r_{1-p})  \geq p \]
which is a contradiction, thus establishing the second assertion.
\end{proof}

The confidence version of the following result essentially completes the proof of Theorem \ref{thm_main}.
\begin{lemma}
\label{lem_main}
  With the assumptions of Theorem \ref{thm_main} 
let $a, a' \in \R$ satisfy $a-a'\geq r_{p} +r_{1-p}$ so that the interval $[
a'+r_{1-p},a-r_{p}]$ is nonempty. Let  $b\in 
 [ a'+r_{1-p},a-r_{p}]$ 
and consider the test $T$ of  ${\cal H}^{\cal D}_{a}$  against ${\cal K}^{\cal D}_{a'}$ 
defined by
\[T \Defi
\begin{cases} 1, & F'(y) \geq b\\
              0, & F'(y) <b \, .
\end{cases}
\]
Then we have 
\[  \theta_{1}(U)  <
e^{-\frac{2( (a
-r_{p})-b)^{2}}{{\cal D}'^{2}}},\]
\[\theta_{2}(U) 
\leq
e^{-\frac{2( b-(a'
+r_{1-p}))^{2}}{{\cal D}'^{2}}}.\]
\end{lemma}
\begin{proof}
Suppose $U \in {\cal H}^{D}_{a}$. Then by Lemma \ref{lem_mean}
we have
\[ b  \leq b-(a -
r_{p})+\E U =  b-(a -
r_{p})+\E F' \]
so that by Theorem \ref{conc_mcdiarmid} applied to $F'$, we conclude that
the
 Type I error satisfies
\[ \theta_{1}(U) =\PR(F'(Y) < b)
\leq
\PR(F'-\E F' < b-(a -r_{p})) <
e^{-\frac{2( (a
-r_{p})-b)^{2}}{D_{F'}^{2}}} \leq e^{-\frac{2( (a
-r_{p})-b)^{2}}{D'^{2}}} .\]
Similarly, if we suppose  $U \in {\cal K}^{D}_{a'}$, by Lemma \ref{lem_mean}
we have
\[ b  \geq b-(a' +
r_{1-p})+\E U=b-(a' + r_{1-p})+\E F' .\] Consequently,
  the Type II error satisfies
\[\theta_{2}(U)  =\PR(F'(Y) \geq b)
\leq
\PR(F'-\E F' \geq b-(a'+ r_{1-p})) \leq
e^{-\frac{2( b-(a'
+r_{1-p}))^{2}}{D_{F'}^{2}}} \leq e^{-\frac{2( b-(a'
+r_{1-p}))^{2}}{D'^{2}}}.\]
\end{proof}
We are now ready to complete the proof of Theorem \ref{thm_main}.
It is easy to show 
 that $a$, $a'$ and $b$ satisfy the assumptions of Lemma
\ref{lem_main}. Since it follows from the assumptions that
$ (a
-r_{p})-b \geq r'_{\d}$ and
$b-(a'
+r_{1-p}) \geq r'_{\d'}$
the assertion follows from Lemma \ref{lem_main}.
\end{proofof}

\begin{proofof}{Corollary \ref{cor_val}}
It is not hard to see that $\E F' =\E F =\E U$.
Let the vector variable $X$ have $J$ components  and index
 the $J^{n}$ components of
$\prod_{i=1}^{n}{{\cal X}_{i}}$ by the map
 $i,j \mapsto k_{ij}=J(i-1)+j,  j=1,..,J, i=1,..,n$.
First observe that if we define
$F':\prod_{i=1}^{n}{{\cal X}_{i}}\rightarrow \R$ by
$F'(\prod_{i=1}^{n}{X_{i}}):= \langle F\rangle_{n}$
 we have $\E F' =\E F =\E U$.
  Moreover, for all $j=1,..,J, i=1,..,n$  we have
$D_{k_{ij}}^{F'} \leq \frac{1}{n} D^{F}_{j}, j=1,..,J , i=1,..,n$
and therefore
\[ {\cal D}_{F'}^{2} =\sum_{j=1,..,J  , i=1,..,n
}{(D^{F'}_{k_{ij}})^{2}}
\leq \frac{1}{n^{2}}\sum_{j=1,..,J  , i=1,..,n
}{(D^{F}_{j})^{2}}
=  \frac{1}{n}\sum_{j=1,..,J
}{(D^{F}_{j})^{2}}
=  \frac{1}{n}  
{\cal D}_{F}^{2}
\]
and so we conclude that
 ${\cal D}_{F'}\leq \frac{ {\cal D}_{F}}{\sqrt{n}} \leq \frac{ {\cal D}}{\sqrt{n}} . $
The assertion follows from  Theorem \ref{thm_main}.
\end{proofof}

\begin{proofof}{Corollary \ref{cor_cert}}
As in the proof of Corollary \ref{cor_val} index
 the $J^{n_{1}+n_{2}}$ components of
$\prod_{i=1}^{n_{1}+n_{2}}{{\cal X}_{i}}$ by the map  $i,j \mapsto k_{ij}=J(i-1)+j,
j=1,..,J, i=1,..,n_{1}+n_{2}$.
Since
$D_{j}^{F} \leq D_{j}^{F_{1}}+D_{j}^{F_{2}},j =1,..,J$ it easily follows the triangle
inequality in $\ell_{2}$ and the assumptions that
 ${\cal D}_{F} \leq {\cal D}_{F_{1}}+{\cal D}_{F_{2}} \leq
{\cal D}_{1}+{\cal D}_{2} \leq {\cal D}$.
Consequently, $U \in {\cal U}_{D}$ and we can apply Theorem
\ref{thm_main}.
 To that end observe that 
$F':\prod_{i=1}^{n_{1}+n_{2}}{{\cal X}_{i}}\rightarrow \R$ defined by
$F'(\prod_{i=1}^{n_{1}+n_{2}}{X_{i}}):= \langle F_{1}\rangle_{n_{1}}+ \langle
F_{2}\rangle_{n_{2}}$ satisfies $\E F' =\E F_{1}+\E F_{2} =\E F =\E U$. Moreover,
it follows that $D_{k_{ij}}^{F'} \leq \frac{1}{n_{1}} D^{F_{1}}_{j}$
 if $1 \leq i  \leq n_{1}$
and 
$D_{k_{ij}}^{F'} \leq \frac{1}{n_{2}}D^{F_{2}}_{j}$ if $n_{1}+1 \leq i  \leq
n_{1}+n_{2}$.
Consequently we conclude that
\begin{eqnarray*}
 {\cal D}_{F'}^{2}& =&\sum_{j=1,..,J  , 1 \leq i  \leq n_{1} 
}{(D^{F'}_{k_{ij}})^{2}}+\sum_{j=1,..,J  , n_{1}+1 \leq i  \leq n_{1}+n_{2} 
}{(D^{F'}_{k_{ij}})^{2}}\\
&\leq& \frac{1}{n_{1}^{2}}\sum_{j=1,..,J  , 1 \leq i  \leq n_{1} 
}{(D^{F_{1}}_{j})^{2}}+ \frac{1}{n_{2}^{2}}\sum_{j=1,..,J  , n_{1}+1
\leq i  \leq n_{1}+n_{2}
}{(D^{F_{2}}_{j})^{2}}
.
\end{eqnarray*}
Therefore, since 
\[ \sum_{j=1,..,J  , 1 \leq i  \leq n_{1}
}{(D^{F_{1}}_{j})^{2}} =n_{1}{\cal D}_{F_{1}}^{2} \leq n_{1}{\cal D}^{2}_{1}\]
and
\[\sum_{j=1,..,J  , n_{1}+1\leq i  \leq n_{1}+n_{2}
}{(D^{F_{2}}_{j})^{2}} =n_{2}{\cal D}_{F_{2}}^{2} \leq n_{2}{\cal D}^{2}_{2}\]
we conclude that
\[  {\cal D}_{F'} \leq
\sqrt{\frac{{\cal D}_{1}^{2}}{n_{1}}+\frac{{\cal D}_{2}^{2}}{n_{2}}}\, .
\]
Consequently
 Theorem \ref{thm_main} implies the first assertion.
For the second observe that ${\cal D}_{1} \leq {\cal D}_{2}$ implies that
${\cal D}_{F} \leq {\cal D}_{1}+{\cal D}_{2} \leq 2{\cal D}_{1}$ which implies $U \in {\cal
U}_{D}$. Moreover,
 setting $n_{2} \geq
\frac{{\cal D}_{2}}{{\cal D}_{1}} n_{1}$ implies that
$ {\cal D}_{F'} \leq  \sqrt{\frac{{\cal D}_{1}^{2}}{n_{1}}+\frac{{\cal D}_{2}^{2}}{n_{2}}}
 \leq  {\cal D}_{1}\sqrt{\frac{2}{n_{1}}}.
$ 
 Theorem \ref{thm_main}
then implies the assertion.
\end{proofof}

\begin{proofof}{Lemma \ref{thm_xval}}
The assumption $d(F,\hat{F}) \leq \d_{p}$ implies that if $\PR(\hat{F} \geq a) \geq p$ that
$\PR(F \geq a) \geq p-\d_{p}$ and if $\PR(\hat{F} \geq a') < p$ that
$\PR(F \geq a') < p+\d_{p}$. The result then follows from Corollary \ref{cor_val}.
\end{proofof}

\begin{proofof}{Lemma \ref{test_basic}}
The first assertion is trivial and since
$g_{H}(\vec{F})+g_{K}(\vec{F}) \leq 0$ the second
follows from
\[ \theta_{2} = \PR( F' > -g_{H}(\vec{F})|{\cal K})\leq \PR( F' \geq
g_{K}(\vec{F})|
{\cal K}) \leq \d_{2}. \]
\end{proofof}

\begin{proofof}{Lemma \ref{c_basic}}
The first assertion follows from the fact that both $F \mapsto {\cal D}_{F}$ and
  $ F \mapsto
 D_{F}$ are diagonal bijective invariants.
The second assertion follows from the fact that both $F \mapsto
{\cal D}_{F}$ and
  $ F \mapsto
D_{F}$  are invariant under $F \mapsto F+b, b \in \R$ and both
transform through scaling $F \mapsto aF$ by ${\cal D}_{aF} =|a|{\cal D}_{F}$
  For the second,
let $x,x'$ approximately achieve the supremum in $  D_{F}$ to accuracy
$\epsilon .$ That is $F(x)-F(x') \geq   D_{F}-\epsilon.$  Then using the
product nature of ${\cal X}$ we find that
\begin{eqnarray*}
F(x)-F(x')& =&F(x_{1},..,x_{m})-F(x'_{1},..,x'_{m})\\
&=&
F(x_{1},x_{2},..,x_{m})-F(x'_{1},x_{2},..,x_{m})+F(x'_{1},x_{2},..,x_{m})-F(x'_{1},x'_{2},..,x_{m})
+......\\
& \leq& \sum_{i=1}^{m}D_{j}^{F}
\end{eqnarray*}
 and so conclude that
$D_{F} \leq \sum_{j=1}^{m}D_{j}^{F}+\epsilon$. Since $\epsilon$ is
arbitrary we then conclude that
\[  D_{F} \leq \sum_{j=1}^{m}D_{j}^{F} \leq
\sqrt{m}\sqrt{\sum_{j=1}^{m}\bigl(D_{j}^{F}\bigr)^{2}}=\sqrt{m}{\cal D}_{F}\] from
which we conclude
that
\[ c_{F}=  \frac{  {\cal D}_{F}}{ D_{F}} \geq \frac{1}{\sqrt{m}} .\]
On the other hand,
since $D_{F} \geq D_{j}^{F}, j=1,..,m$ we obtain
\[  D_{F}^{2} \geq \frac{1}{m}\sum_{j=1}^{m}{\bigl(
D_{j}^{F}\bigr)^{2}}=\frac{1}{m}{\cal D}_{F}^{2}
\]
and conclude that
$ c_{F}= \frac{  {\cal D}_{F}}{ D_{F}} \leq \sqrt{m}.$

\end{proofof}

\begin{proofof}{Theorem \ref{thm_serfling}}
Consider the first set of assertions.
According to the proof of \cite[Thm.~2.3.2]{serfling} we have
\[\PR^{n}(\hat{\xi}_{q}>\xi)=\PR^{n}\Bigl(\sum_{i=1}^{n}V_{i}-\sum_{i=1}^{n}\E V_{i} >
n\d_{1}\Bigr)\]
where
$\d_{1} := \F(\xi)-q$ and $V_{i}:= I(X_{i} > \xi)$. Consequently we obtain
$\E V_{i}=1-\F(\xi), i=1,..,n.$ Applying Hoeffding's inequality
\cite[Eqn.~2.4]{McDiarmid91}
establishes the first assertion. The second assertion follows from
\cite[Thm.~2.3b]{McDiarmid91}
which he attributes to
\cite{AnVa79} in the binomial case. The last assertion follows from
\cite[Thm.~2.3c]{McDiarmid91} by the change of variables $V_{i}' := 1-V_{i}$,
which he also attributes to
\cite{AnVa79} in the binomial case.

For the second set of assertions, observe that according to the proof of
\cite[Thm.~2.3.2]{serfling} we have
\[\PR^{n}(\hat{\xi}_{q}<\xi)\leq \PR^{n}\Bigl(\sum_{i=1}^{n}W_{i}-\sum_{i=1}^{n}\E W_{i} >
n\d_{2}\Bigr)\]
where
$\d_{2} := q- \F(\xi)$ and $W_{i}:= I(X_{i} \leq \xi)$.
Consequently we obtain $\E W_{i}=\F(\xi), i=1,..,n.$ The assertions then follow
in the same way as in the first set but with the role of $\F$ switched with
$1-\F$.
\end{proofof}

\begin{proofof}{Theorem \ref{thm_serfling2}}

Since the first  assertion is  
 is clearly true when $\xi_{p}-\xi_{1-p}\leq 0$  we can 
 assume $\xi_{p}-
\xi_{1-p} > 0$ and $\frac{1}{2} < p <1$.
First observe that for $c,\a \in \R$ we have
\begin{equation}
\label{i1}
\PR^{n}\Bigl(\hat{\xi}_{p'}-\hat{\xi}_{1-p'}<c(\xi_{p}-\xi_{1-p})\Bigr)
\leq
\PR^{n}\Bigl(\hat{\xi}_{p'}<c\xi_{p}+\a\Bigr)+\PR^{n}\Bigl(\hat{\xi}_{1-p'}>c\xi_{1-p}+\a\Bigr)\,
.
\end{equation}
We will address each term on the right-hand side separately using
Theorem
\ref{thm_serfling}. For $\epsilon > 0$ let $p' > 1-\frac{1}{n}$  so that we
have the identities
$\hat{\xi}_{p'}= \sup_{i=1,..,n}{X_{i}}$ and
$\hat{\xi}_{1-p'}=\sup_{i=1,..,n}{X_{i}}$.
Since $ p' > p$ it follows that $\xi_{p}\leq  \xi_{p'}$  and
consequently 
  $\xi_{p}-\epsilon < \xi_{p'}$. Moreover, a similar argument
shows that $\xi_{1-p}+\epsilon > \xi_{1-p'}$.
Consequently, if
we define
\[c_{\epsilon}\Defi 1-
\frac{2\epsilon}{\xi_{p}-\xi_{1-p}}
\] and
\[ \a_{\epsilon} \Defi
\frac{\xi_{p}(\xi_{1-p}+\epsilon)-\xi_{1-p}(\xi_{p}-\epsilon)}{\xi_{p}-\xi_{1-p}}\]
it follows  that $c_{\epsilon} < 1  $ and
\[ c_{\epsilon}\xi_{p}+\a_{\epsilon}=\xi_{p}-\epsilon<\xi_{p'} \]
\[c_{\epsilon}\xi_{1-p}+\a_{\epsilon} = \xi_{1-p}+\epsilon  > \xi_{1-p'}.\]
 Consequently we can apply Part 2iii of Theorem \ref{thm_serfling} to the
first term on the right-hand side of (\ref{i1}) to obtain
\[\PR^{n}\Bigl(\sup_{i=1,..,n}{X_{i}}<c_{\epsilon}\xi_{p}+\a_{\epsilon}\Bigr)=
\PR^{n}\Bigl(\sup_{i=1,..,n}{X_{i}}<\xi_{p} -\epsilon\Bigr)=\PR^{n}\Bigl(\hat{\xi}_{p'}<\xi_{p} -\epsilon\Bigr)\leq
e^{-\frac{n\d_{2}^{2}}{2(1-\F(\xi_{p} -\epsilon))}}
\]
where $\d_{2}:=  p'-\F(\xi_{p} -\epsilon)$.
Letting $p' \mapsto 1$ we obtain
\[\PR^{n}\Bigl( \sup_{i=1,..,n}{X_{i}}< c_{\epsilon}\xi_{p}+\a_{\epsilon}\Bigr)\leq
e^{-\frac{1}{2}n\bigl(1-\F(\xi_{p} -\epsilon)\bigr) }.
\]
Since $\F(\xi_{p} -\epsilon) \leq \F(\xi_{p}-) \leq p$ we then conclude
that 
\begin{equation}
\label{e1}
\PR^{n}\Bigl( \sup_{i=1,..,n}{X_{i}}<c_{\epsilon}\xi_{p}+\a_{\epsilon}\Bigr)\leq
e^{-\frac{1}{2}n(1-p) }.
\end{equation}

 For the second term on the right of (\ref{i1}) we apply
Part 1iii of Theorem \ref{thm_serfling} to obtain
\[\PR^{n}\Bigl(
\inf_{i=1,..,n}{X_{i}}>
c_{\epsilon}\xi_{1-p}+\a_{\epsilon}\Bigr)=\PR^{n}\Bigl(
\inf_{i=1,..,n}{X_{i}}>\xi_{1-p}+\epsilon\Bigr)=
\PR^{n}\Bigl(\hat{\xi}_{1-p'}>\xi_{1-p}+\epsilon\Bigr)\leq
e^{-\frac{n\d_{1}^{2}}{2\F(\xi_{1-p}+\epsilon)}}
\]
where $\d_{1}:=  \F(\xi_{1-p}+\epsilon)-1+p'$. Letting $p'\mapsto 1$
and using $F(\xi_{1-p}+\epsilon)\geq F(\xi_{1-p})\geq 1-p$ 
 we obtain
\begin{equation}
\label{e2}
\PR^{n}\Bigl( \inf_{i=1,..,n}{X_{i}}<c_{\epsilon}\xi_{p}+\a_{\epsilon}\Bigr)\leq
e^{-\frac{1}{2}n(1-p) }.
\end{equation}
We combine the inequalities (\ref{e1}) and (\ref{e2}) with (\ref{i1}) to
obtain
\[
\PR^{n}\Bigl(D_{n}<c_{\epsilon}(\xi_{p}-\xi_{1-p})\Bigr)
\leq 2e^{-\frac{1}{2}n(1-p) }.
\]
Since $c_{\epsilon} \uparrow 1$ as $\epsilon \downarrow 0$ the first
assertion of the theorem follows (see e.g.~ \cite[Thm.~1.2.7]{Ash72}).

For the second assertion, 
 observe that it is clearly 
 true when $\xi_{p}=0$ so we can assume  
 $\xi_{p}>  0.$  Now observe that the proof of 
Equation (\ref{e1}) actually proved that 
\[ \PR^{n}\Bigl(
\sup_{i=1,..,n}{X_{i}}<\xi_{p} -\epsilon\Bigr)\leq
e^{-\frac{1}{2}n(1-p) }.\]
Since $(\xi_{p} -\epsilon) \uparrow \xi_{p}$ as $\epsilon \downarrow 0$
the second assertion follows.

\end{proofof}

\begin{proofof}{Proposition \ref{d_suff}}
Let $\d_{+} := 1-\F(X_{1}-\frac{\e}{2(1+\e)} { D})$ and
 $\d_{-} := \F(X_{-}+\frac{\e}{2(1+\e)} { D})$ and define
$\d^{*}=\min{(\d_{-},\d_{+})}-\d' $ with $\d' >0.$ Then since  Lemma
\ref{lem_F}asserts that
$\F^{-1}(t) \leq x$ if and only if $t \leq \F(x)$ it follows that
$\F^{-1}(\F(X_{+}-\frac{\e}{2(1+\e)}
{D})+\d') > X_{+}-\frac{\e}{2(1+\e)} { D}$ and therefore
\[\x_{1-\d^{*}}\geq \x_{1-\d_{+}+\d'}
=\F^{-1}(1-\d_{+}+\d')=\F^{-1}(\F(X_{+}-\frac{\e}{2(1+\e)}
{ D})+\d') \geq X_{+}-\frac{\e}{2(1+\e)} { D}.\]
Moreover, since
\[\x_{\d^{*}}\leq  \xi(\d_{-})
= \F^{-1}(\d_{-})=\F^{-1}(\F(X_{-}+\frac{\e}{2(1+\e)} { D}))
\leq X_{-}+\frac{\e}{2(1+\e)} { D}\] we conclude that
\[\x_{1-\d^{*}}-\xi_{\d^{*}} \geq\bigl(1-\frac{2\e}{2(1+\e)}\bigr) {
D}=\frac{{ D}}{1+\e}.
\]
The assertion then follows by letting $\d' \mapsto 0$.
\end{proofof}

\begin{proofof}{Corollary \ref{cor_diam2}}
We will only prove the first assertion since the proof of the second is
essentially the same.
Since the assertion is trivially true when
$ D=0$ we can assume this not the case.
Then since for $  0 <p \leq \d(\epsilon)$ we have
\[\xi_{1-p}-\xi_{p}= \frac{\xi_{1-p}-\xi_{p}}{{ D}} D\geq \frac{{ D}}{1+\epsilon}\]
the assertion  follows from
from Theorem \ref{thm_serfling2}.
\end{proofof}


\begin{proofof}{Theorem \ref{thm_test}}
Since $\hat{D} \leq D$ with probability $1$ we find that
\begin{eqnarray*}
\theta_{1}&=&\PR^{n}\bigl(T_{1}
=0\big|{\cal H}_{1}\bigr)\\
&=&\PR^{n}\bigl(f_{H}((1+\epsilon)\hat{D})+f_{K}((1+\epsilon)\hat{D}) >
0\big|f_{H}((1+\epsilon)D)+f_{K}((1+\epsilon)D)\leq 0\bigr)\\
&\leq & \PR^{n}\bigl(f_{H}((1+\epsilon)D)+f_{K}((1+\epsilon)D) >
0\big|f_{H}((1+\epsilon)D)+f_{K}((1+\epsilon)D)\leq 0\bigr)\\
& =&0
\end{eqnarray*}
establishing the first assertion.
 For the second and third assertions we use following lemma.
\begin{lemma}
\label{lem_sample}

With the assumptions of Theorem \ref{thm_test}
 let  $ F':{\cal X}^{n} \rightarrow \R$ and
let  $f:\R^{k} \rightarrow \R$ be non-decreasing.
Then we have
\begin{eqnarray}
\label{unionbound}
 \PR^{n}\Bigl(F' \geq  f\bigl((1+\epsilon)\hat{ D}
\bigr)
 \Bigr)
& \leq &
\PR^{n}\Bigl(  F' \geq f(D )
\Bigr)+ \sum_{j=1}^{k}\PR^{n_{j}}\Bigl( (1+\epsilon)\hat{ D}_{F^{j}} < {
D}_{F^{j}}\Bigr)\, .
\end{eqnarray}
Moreover,
 for  $  n_{j} \geq  n_{j}^{\epsilon}(\d), j=1,..,k $  we have
\[  \sum_{j=1}^{k}\PR^{n}\Bigl( (1+\epsilon)\hat{ D}_{F^{j}} < {
D}_{F^{j}}\Bigr) \leq \d\]
and therefore
\begin{equation*}
 \PR^{n}\Bigl(F'  \geq f\bigl((1+\epsilon)\hat{ D}
 \bigr)\Bigr) \leq \PR^{n}\Bigl(F'  \geq f\bigl( D
 \bigr)\Bigr) + \d  \,
\end{equation*}

\end{lemma}
\begin{proof}
By the monotonicity of $f$ it follows that
\begin{eqnarray}
\label{union}
&& \Bigl\{F- f\bigl((1+\epsilon)  \hat{D}
\bigr) \geq 0\Bigr\}\\
 &  \subset &
 \Bigl\{ F-f( D) \geq
0\Bigr\} \cup \cup_{j=1}^{k}
 \Bigl\{(1+\epsilon)\hat{ D}_{F^{j}} <  D_{F^{j}} \Bigr\}\, .
\end{eqnarray}
Consequently, we obtain the first assertion:
\begin{eqnarray*}
\PR^{n}\Bigl(F' \geq  f\bigl((1+\epsilon)\hat{ D}
\bigr)
 \Bigr)
 &\leq&
\PR^{n}\Bigl(  F' \geq f(D )
\Bigr)+ \sum_{j=1}^{k}\PR^{n}\Bigl( (1+\epsilon)\hat{ D}_{F^{j}} < {
D}_{F^{j}}\Bigr)\\
& = &\PR^{n}\Bigl(  F' \geq f(D )
\Bigr)+  \sum_{j=1}^{k}\PR^{n}\Bigl( (1+\epsilon)\hat{ D}_{F^{j}} < {
D}_{F^{j}}\Bigr)\, .
\end{eqnarray*}
For the second, observe that by Corollary
\ref{cor_diam2} we find that for each $j$ we have
\[\PR^{n}\Bigl( (1+\epsilon)\hat{
D}_{F^{j}} < {
D}_{F^{j}}\Bigr) =\PR^{n}\Bigl( (1+\epsilon)\hat{
D}_{F^{j}} < {
D}_{F^{j}}\Bigr)  \leq 2e^{-\frac{n\t^{j}(\epsilon)}{2}}.\] Since
the
assumption $n\geq  n_{j}^{\epsilon}(\d)$, defined in
(\ref{n_def}),
 implies that
\[2e^{-\frac{n\t^{j}(\epsilon)}{2}} \leq \frac{ \d}{k}, \quad
j=1,..,k \] the
second assertion follows from the first.
\end{proof}

We now proceed to the second and third assertions of Theorem \ref{thm_test}.
Observe that
\begin{eqnarray*}
\theta_{11}&=&\PR^{n}\bigl(T_{1}
=1,T_{2}=0\big|{\cal H}_{2\epsilon}\bigr)\\
&\leq&\PR^{n}\bigl(
T_{2}=0\big|{\cal H}_{2}\bigr)\\
&=&\PR^{n}\bigl(
F' \leq -f_{H}((1+\epsilon)\hat{D})\big|{\cal H}_{2}\bigr)
\end{eqnarray*}
Since $ \Delta = \sum_{j=1}^{k}\PR^{n}\Bigl( (1+\epsilon)\hat{
D}_{F^{j}} < D_{F^{j}}\Bigr)$,
Lemma \ref{lem_sample} (Equation \ref{unionbound}) applied to $-F'$ then shows that
\[\theta_{11} \leq \PR^{n}\bigl(
F' \leq -f_{H}((1+\epsilon)\hat{D})\big|{\cal H}_{2}\bigr) \leq  \PR^{n}\bigl(
F' \leq -f_{H}(D)\big|{\cal H}_{2}\bigr) +\Delta\] thus establishing the second
assertion.
Since $T_{1}=1 $ implies that $f_{H}((1+\epsilon)\hat{D})+f_{K}((1+\epsilon)\hat{D})
\leq 0$ we find that
\begin{eqnarray*}
\theta_{12}&=&\PR^{n}\bigl(T_{1}
=1,T_{2}=1\big|{\cal K}_{2\epsilon}\bigr)\\
&=&\PR^{n}\bigl(T_{1}
=1, F'  > -f_{H}((1+\epsilon)\hat{D})\big|{\cal K}_{2\epsilon}\bigr)\\
&\leq&
\PR^{n}\bigl(T_{1}
=1, F'  > f_{K}((1+\epsilon)\hat{D})\big|{\cal K}_{2\epsilon}\bigr)\\
&\leq&\PR^{n}\bigl(
F' \geq f_{K}((1+\epsilon)\hat{D})\big|{\cal K}_{2}\bigr)
\end{eqnarray*}
As in the previous case, Lemma \ref{lem_sample} then shows that
\[\theta_{12} \leq \PR^{n}\bigl(
F' \geq f_{K}((1+\epsilon)\hat{D})\big|{\cal K}_{2}\bigr) \leq  \PR^{n}\bigl(
F' \geq f_{K}(D)\big|{\cal K}_{2}\bigr) +\Delta\] thus establishing the third
assertion.

The last set of assertions follows by  observing that the assumption $  n \geq
\max{\bigl( n_{j}^{\epsilon}(\d_{1}),
n_{j}^{\epsilon}(\d_{2})\bigr)}, j=1,..,k $ and Lemma \ref{lem_sample} implies that
$\Delta \leq \min{(\d_{1},\d_{2})}$.

\end{proofof}

\begin{proofof}{Corollary \ref{cor_val_e}}
Since $\E F' =\E F$, Lemma \ref{lem_inversion} implies that
\[\PR^{n}\Bigl(F'  \leq - f'_{H}({\cal D}_{F}, {\cal D}_{F'},\d_{1})
\big|{\cal H}_{2} \Bigr) \leq \d_{1}\, , \]
\[\PR^{n}\Bigl(F'  \geq f'_{K}({\cal D}_{F}, {\cal D}_{F'},\d_{2})\big|
{\cal K}_{2}\Bigr) \leq \d_{2}\, , \]
where
\[ f'_{H}(r_{1},r_{2},\d)\Defi
\frac{r_{1}}{\sqrt{2}}\sqrt{\log{\d^{-1}}}+
\frac{r_{2}}{\sqrt{2}}\sqrt{\log{p^{-1}}}-a\, ,\]
\[ f'_{K}(r_{1},r_{2},\d)\Defi
\frac{r_{1}}{\sqrt{2}}\sqrt{\log{\d^{-1}}}+
\frac{r_{2}}{\sqrt{2}}\sqrt{\log{(1-p)^{-1}}}+a'\,.\]
By Definition \ref{cf} we have
${\cal D}_{F}\leq cD$. Moreover,
  the proof of Corollary \ref{cor_val}
  shows that
\[{\cal D}_{F'}\leq \frac{1}{\sqrt{n}}{\cal D}_{F} \leq
\frac{cD}{\sqrt{n}}. \]
Consequently, we have 
$f'_{H}({\cal D}_{F}, {\cal D}_{F'},\d_{1}) \leq f_{H}(D)$ and
$f'_{K}({\cal D}_{F}, {\cal D}_{F'},\d_{2}) \leq f_{K}(D)$   and
therefore
\[\PR^{n}\Bigl(F'  \leq - f_{H}(D)
\big|{\cal H}_{2} \Bigr) \leq \d_{1}\, , \]
\[\PR^{n}\Bigl(F'  \geq f_{K}(D)\big|
{\cal K}_{2}\Bigr) \leq \d_{2}\, , \]
The assertion then follows from Theorem \ref{thm_test}.

\end{proofof}

\begin{proofof}{Corollary \ref{cor_cert_e}}
As in the proof of Corollary \ref{cor_val_e}, since $\E F' =\E F$, Lemma
\ref{lem_inversion} implies that
\[\PR^{n}\Bigl(F'  \leq - f'_{H}({\cal D}_{F}, {\cal D}_{F'},\d_{1})
\big|{\cal H}_{2} \Bigr) \leq \d_{1}\, , \]
\[\PR^{n}\Bigl(F'  \geq f'_{K}({\cal D}_{F}, {\cal D}_{F'},\d_{2})\big|
{\cal K}_{2}\Bigr) \leq \d_{2}\, , \]
where
\[ f'_{H}(r_{1},r_{2},\d)\Defi
\frac{r_{1}}{\sqrt{2}}\sqrt{\log{p^{-1}}}+
\frac{r_{2}}{\sqrt{2}}\sqrt{\log{\d^{-1}}}-a\, ,\]
\[ f'_{K}(r_{1},r_{2},\d)\Defi
\frac{r_{1}}{\sqrt{2}}\sqrt{\log{(1-p)^{-1}}}+
\frac{r_{2}}{\sqrt{2}}\sqrt{\log{\d^{-1}}}+a'\,.\]
By Definition \ref{cf} we have
${\cal D}_{F_{1}}\leq c_{1}D_{1}$ and ${\cal D}_{F_{2}}\leq
c_{2}D_{2}$. Therefore it follows that
\[{\cal D}_{F}={\cal D}_{F_{1}+F_{2}} \leq {\cal
D}_{F_{1}}+{\cal D}_{F_{2}} \leq c_{1}D_{1}+c_{2}D_{2}.\]
 Moreover, the proof of
  Corollary
 \ref{cor_cert} implies that
\[{\cal D}_{F'}\leq \sqrt{\frac{{\cal D}^{2}_{F_{1}}}{n_{1}}
+\frac{{\cal
D}^{2}_{F_{2}}}{n_{2}}}  \leq \sqrt{\frac{c^{2}_{1}D^{2}_{1}}{n_{1}}
+\frac{c_{2}^{2}
D^{2}_{2}}{n_{2}} }
. \]
Consequently, we have
$f'_{H}({\cal D}_{F}, {\cal D}_{F'},\d_{1}) \leq f_{H}(D)$ and
$f'_{K}({\cal D}_{F}, {\cal D}_{F'},\d_{2}) \leq f_{K}(D)$   and
therefore
\[\PR^{n}\Bigl(F'  \leq - f_{H}(D)
\big|{\cal H}_{2} \Bigr) \leq \d_{1}\, , \]
\[\PR^{n}\Bigl(F'  \geq f_{K}(D)\big|
{\cal K}_{2}\Bigr) \leq \d_{2}\, , \]
The assertion then follows from Theorem \ref{thm_test}.

\end{proofof}

\section{Appendix}

The following Lemma from \cite[Lem.~1.1.4 \& Sec.~2.3]{serfling}
 lists important properties
of the distribution function $\F(x) := \PR(X \leq x)$ and its corresponding
 quantile function
$\F^{-1}(t) := \inf{\{x: \F(x) \geq t\}}$.
\begin{lemma}
\label{lem_F}
Let  $\F$ be a distribution function. Then $\F$ is  right continuous and the
function $\F^{-1}, 0 < t< 1$ is non-decreasing, left continuous and satisfies
\begin{enumerate}
\item $\F^{-1}(\F(x)) \leq x,  -\infty < x < \infty \,.$
\item $\F(\F^{-1}(t)) \geq t \geq \F(\F^{-1}(t)-),  0 < t < 1\, .$
\item $\F(x) \geq t$ if and only if $x \geq \F^{-1}(t).$
\end{enumerate}
\end{lemma}

\begin{example}[Extreme values of $c_{F}$]
\label{ex_c}
Let $F(x) := \sum_{j=1}^{m}{F_{j}(x_{j})}$. Then
since
$F(x)-F(x')=\sum_{j=1}^{m}{\bigl(F_{j}(x_{j})-F_{j}(x'_{j})\bigr)}$
it follows that \[ D_{F}=
\sum_{j=1}^{m}{D_{j}^{F_{j}}}.\]
Moreover, since
$D^{F}_{j}=D^{F_{j}}_{j}, j=1,..,m$ we obtain
\[
{\cal D}^{2}_{F}=\sum_{j=1}^{m}{\bigl(D^{F_{j}}_{j}\big)^{2}}\]
and therefore
\[ c^{2}_{F} = \frac{ \sum_{j=1}^{m}{\bigl(D^{F_{j}}_{j}\bigr)^{2}}
}{\Bigl(\sum_{j=1}^{m}{D^{F_{j}}_{j}}\Bigr)^{2}}.\]
In particular, when $D^{F_{j}}_{j}=D^{F_{1}}_{1}, j=1,,m$ we obtain
$  c_{F}=\frac{1}{\sqrt{m}}.$
On the other hand, let $ {\cal X}:= [0,1]^{m} \subset \R^{m}$ 
and  let $F(x) := \snorm{x}, \snorm{x}\leq 1 $ and
$F(x) := 0, \snorm{x}> 1 $ where $\snorm{x}$ is the Euclidean
norm of $x$.
 Then it is easy to
see
that $ D_{F}=1$, $D_{j}^{F}=1, j=1,..,m$ and therefore ${\cal D}_{F}^{2}=
m$.
Consequently  in this case we obtain
$ c_{F} = \sqrt{m}.$ 
\end{example}

\section*{Acknowledgments}
We gratefully acknowledge several illuminating discussions with Nicholas
Hengartner regarding this work.